\documentclass[conference]{IEEEtran}
\usepackage{amsmath, amsthm, amssymb, amsfonts, bm}
\usepackage{graphicx}
\usepackage{color}
\usepackage{algorithm}
\usepackage{algpseudocode}
\usepackage{algpseudocodex}
\usepackage{color}

\usepackage{optidef}
\usepackage{authblk}

\usepackage{algpseudocode}
\usepackage{caption}
\newtheorem{theorem}{Theorem}
\newtheorem{statement}{Statement}
\newtheorem{lemma}{Lemma}

\newtheorem{remark}{Remark}

\title{A Stochastic Geometry Based Techno-Economic Analysis of RIS-Assisted Cellular Networks}

\author[1,2,3]{Guodong~Sun}
\author[1,2]{Fran\c{c}ois~Baccelli}
\author[3]{Luis~Uzeda~Garcia}
\author[3]{Stefano~Paris}
\affil[1]{Department d'informatique, Ecole Normale Sup\'erieure. Email: guodong.sun$|$francois.baccelli@ens.fr.}
\affil[2]{Institut national de recherche en sciences et technologies du num\'erique (INRIA).}
\affil[3]{Nokia Networks France. Email: luis.uzeda\_garcia$|$stefano.paris@nokia.com.}

\begin{document}

\maketitle

\begin{abstract}
Reconfigurable intelligent surfaces (RISs) are a promising technology for enhancing cellular network performance and yielding additional value to network operators. 
This paper proposes a techno-economic analysis of RIS-assisted cellular networks to guide operators in deciding between deploying additional RISs or base stations (BS).
We assume a relative cost model that considers the total cost of ownership (TCO) of deploying additional nodes, either BSs or RISs. 
We assume a return on investment (RoI) that is proportional to the system's spectral efficiency. 
The latter is evaluated based on a stochastic geometry model that gives an integral formula for the ergodic rate in cellular networks equipped with RISs. 
The marginal RoI for any investment strategy is determined by the partial derivative of this integral expression with respect to node densities.
We investigate two case studies: throughput enhancement and coverage hole mitigation. These examples demonstrate how operators could determine the optimal investment strategy in scenarios defined by the current densities of BSs and RISs, and their relative costs.
Numerical results illustrate the evolution of ergodic rates based on the proposed investment strategy, demonstrating the investment decision-making process while considering technological and economic factors.
This work quantitatively demonstrates that strategically investing in RISs can offer better system-level benefits than solely investing in BS densification. 
\end{abstract}

\section{Introduction}

The reconfigurable intelligent surface (RIS) technology is considered a promising complement to base stations (BS) and is expected to enhance a wide range of wireless services, including extended coverage, higher data rates, and improved reliability~\cite{xu2020understanding}. 
A RIS is composed of numerous low-cost reflectors that can intelligently shape radio signals with minimal power consumption, as only configuration energy is required.  
This makes the capital expenditure of deploying RIS hardware significantly lower than other beamforming technologies that rely on active antenna arrays~\cite{liu2022path}.
RISs can also reduce networks' operation expenditure since they consume less power compared to BSs~\cite{wang2021joint}. 
Additionally, since RISs have no power amplification, they generate less interference than BSs, and this feature also contributes to power efficiency and cost savings. 
Overall, the RIS technology may offer a cost-effective alternative to deploying additional BSs in cellular networks.

While many researchers have addressed the technical challenges of incorporating RISs into cellular networks, few results are available to assess economic questions quantitatively. 
In the literature on relay technology, the system-level techno-economic analyses mostly relied on Monte Carlo simulation results to quantify the expenditure for different technologies~\cite{lang2009business}. 
An analytical approach would be most helpful in quantifying the techno-economic trade-offs associated with RIS deployment.

Stochastic geometry serves as a powerful analytical tool for modeling and evaluating large-scale networks, effectively capturing the inherent uncertainty in user location and irregularity in node placement~\cite{baccelli2010stochastic}.
This tool has been applied to assess the performance of RIS.
For example, the authors in \cite{lyu2021hybrid} model both BSs and RISs as homogeneous Poisson point processes (PPP) and associate the typical user equipment (UE) to the nearest BS and RIS. 
By approximating the composite signal with a Gamma distribution, the authors derive a double integral expression to quantify the coverage probability. 
To consider the spatial correlation between RISs and BSs, the authors in~\cite{wang2023performance} use the Gauss-Poisson process to model the BSs and RISs, with a fixed distance. They also employ the Gamma approximation framework to analyze the performance. 
Recently, the authors in~\cite{sun2023performance} proposed a novel framework that captures the random placement of both BSs and RISs, enabling the performance analysis of a variety of RIS applications, such as throughput enhancement or coverage hole mitigation.
The authors in~\cite{sun2023performance} provide integral expressions for computing ergodic rate analytically.

In this work, we address the analysis of RISs from a techno-economic perspective using the stochastic geometry tools developed in \cite{sun2023performance} to characterize the ergodic rate of a typical UE, which reflects the spectral efficiency at the system level. 
To assess the marginal spectral efficiency of an investment, we introduce a new sensitivity analysis step into the integral expressions obtained in~\cite{sun2023performance}. 
One of our contributions is the verification of conditions for interchanging integration and differentiation operations. 
More precisely, the present paper provides integral formulas to compute the partial derivatives
of the ergodic rate of a typical UE with respect to (w.r.t.) the density of BSs and that of RISs.
By applying these formulas to a given cellular network with RISs, we can quantify the expected gain from increasing BS or RIS densities. This could be used to guide operators in deciding whether to invest in additional BSs or RISs so as to optimize the total cost of ownership (TCO). 
This investment strategy requires knowledge of the current densities of BSs and RISs and their relative costs, assuming a simplified return on investment (RoI) proportional to the ergodic rate of a typical UE.

This work is organized as follows: we present
the system model of the RIS-assisted cellular network and discuss the TCO estimation in Section~\ref{section:system_model}. 
We derive the first-order derivative expressions of the system performance w.r.t. the densities of BSs and RISs,
and discuss the optimal investment strategy in Section~\ref{section:analysis}.  
Numerical results are presented in~Section~\ref{sec:simulation}, where we interpret the techno-economic analysis.
We conclude this paper in Section~\ref{sec:conclusion}.

\section{Stochastic Geometry System Model}\label{section:system_model}

RISs can be deployed to address two primary use cases: enhancing network throughput and mitigating coverage holes. We present the stochastic geometry models of these two scenarios, accounting for the randomness of wireless network deployments.
We also model the TCO for deploying BSs or RISs per unit area.

\subsection{Thoughput Enhancement Model}
\begin{figure}
    \centering
    \includegraphics[width=0.7\linewidth]{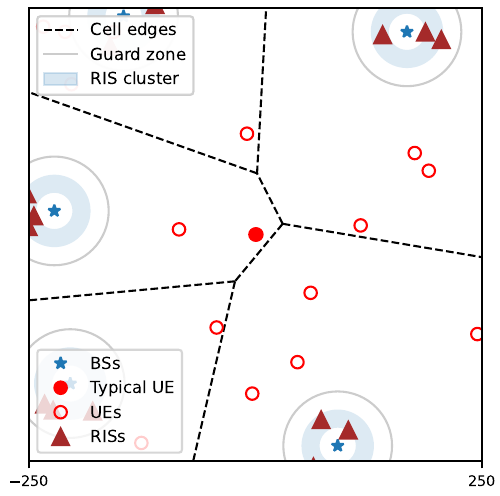}
    \caption{Modeling RIS-assisted Cellular Networks with Mat\'ern Cluster Processes, with a guard zone to prevent the typical UE from being near BSs or RISs.}
    \vspace{-0.5cm}
    \label{fig:system_model}
\end{figure}
We investigate the throughput enhancement scenario by modeling BSs and RISs as Mat\'ern cluster processes (MCP), as illustrated in Fig.~\ref{fig:system_model}. 
The location of BSs is modeled as a homogeneous PPP
$\Phi_{\rm BS} \triangleq \sum_{i\in \mathcal I} \delta_{\mathbf{x}_i}$
with intensity $\lambda_{{\rm BS}}$, where $\delta$ denotes the Dirac delta function, modeling the mother process of the MCP.  
Here, $\mathcal{I}$ is the index set of BSs.
Each RIS is assumed to be associated with a serving BS and the placement of RISs is modeled as the daughter processes of the MCP
$\Phi_{\rm RIS} = \cup_{i\in \mathcal{I}} \phi_i$, where 
$\phi_i$ denotes the RIS cluster of BS $i$.
By definition, $\phi_i= \sum_{j \in \mathcal{J}_i} \delta_{\mathbf{y}_{i,j}}$, 
where $\mathcal{J}_i$ is the index set of RISs of cluster $i$, is a PPP with intensity $\lambda_{{\rm RIS}}$,
with support on the ring  $\mathbb{D}_{\mathbf{x}_i}(R_{\rm in}, R_{\rm out})$ of
center $\mathbf{x}_i$, inner radius $R_{\rm in}$ and outer radius $R_{\rm out}$.
In addition, the UE locations are modeled as a homogeneous PPP $\Phi_{{\rm UE}}$ with intensity
$\lambda_{{\rm UE}}$, which is independent of $\Phi_{\rm BS}$ and $\Phi_{\rm RIS}$. 
We consider a guard zone to exclude UEs near BSs from throughput enhancement evaluation, as strong signals experienced by these UEs reduce the need for RIS configuration and overhead. 
Therefore, we only consider UEs beyond a certain distance $R_c$ to be jointly served by the BS and RISs in this cell, as indicated by the gray contour in Fig.~\ref{fig:system_model}. 
The typical UE in $\Phi_{{\rm UE}}$, marked as the filled red circle, is located at the origin $o$ and is served by its closest BS and the associated RIS cluster that is specifically indexed by $o$.
The probability distribution function (PDF) of the distance from the serving BS to the typical UE is $f(r)=2\pi \lambda_{{\rm BS}}re^{- \lambda_{{\rm BS}} \pi r^2}$\cite{andrews2011tractable}\cite{baccelli2020random}, with $r \triangleq ||\mathbf{x}_o||$.

\subsection{Coverage Hole Mitigation Model}
\begin{figure}
    \centering
    \includegraphics[width=0.7\linewidth]{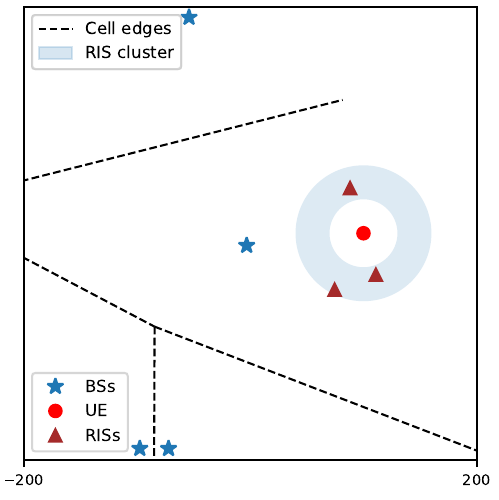}
    \caption{RISs are deployed to assist UEs in coverage holes}
    \label{fig:coverage_hole_model}
    \vspace{-0.5cm}
\end{figure}

To investigate coverage hole mitigation, we consider a scenario where a UE located at the center of a coverage hole experiences signal attenuation on the direct link.
We consider a cluster of RISs $\mathcal{J}_{o}$ is deployed in a ring around this hole, as shown in Fig.~\ref{fig:coverage_hole_model}. 
To model the effect of blockage on signal attenuation, we assume the direct link from the BS to the UE is penalized by a constant attenuation coefficient $K$. 
To model the distance between the coverage hole and the associated BS, we remark that the spatial distribution of BSs will impact the proximity of an arbitrary location to its closest BS~\cite{chiu2013stochastic}. 
For the two-dimensional PPP-modeled BSs, the expected value of this distance is inversely proportional to the square root of  $\lambda_{\rm BS}$. 
In this work, we consider an average scenario where the centers of both the coverage hole and the RIS ring are located at a distance $r_{\rm H}$ from its associated BS, which is a function of the BS density, given by 
\vspace{-0.1cm}
\begin{equation}\label{eq:cell_center}
	r_{\rm H} = \frac{C_{\rm H}}{\sqrt{\lambda_{\rm BS}}},
    \vspace{-0.1cm}
\end{equation}
where $C_{\rm H}$ is a constant, whose value is selected in the simulation section ensuring that the coverage hole and the cluster of RISs are with high probability within the cell.

Without loss of generality, we assume that the RISs can effectively serve their associated UEs\footnote{Models of reflective, refractive, or hybrid (reflective and refractive) RISs can be analyzed using the same method by changing the density of the corresponding RIS type.}.
In other words, each RIS is modeled as a point and is capable of reflecting signals in any direction according to a given RIS-UE association.
Note that both models contain the special case of the scenario without deploying RISs by setting $\lambda_{\rm RIS}=0$. 

\subsection{System Performance Analysis}

The channel from the BS to the UEs assisted by the RISs is time-dispersive. This is because RISs are spatially dispersed and the propagation distance of the direct and reflected paths are different.
It is shown in \cite{sun2023performance} that employing OFDM to mitigate the time-dispersion of
the channel caused by multipath reflections, the post-processed SINR is given by 
\begin{equation}\label{eq:SINR-def}
\text{SINR} = \frac{|\gamma_{D_o}|^2 + \sum_{j \in \mathcal{J}_o} | \gamma_{{R}_{o,j}}|^2 }{\sum_{i\in \mathcal{I}\setminus o}\big( |\gamma_{D_i}|^2 + \sum_{j \in \mathcal{J}_i} | \gamma_{{R}_{i,j}}|^2  \big) + \sigma^2},
\end{equation}
where $\sigma^2$ is the power of the additive Gaussian noise.
Here, $\gamma_{D_o}$ and $\gamma_{R_{o,j}}$ represent the received signal amplitude of the direct and reflected intended signals, respectively, while $\gamma_{D_i}$ and $\gamma_{{R}_{i,j}}$ represent the corresponding received signal amplitude of the direct and reflected interference, respectively.
Each signal amplitude is typically the product of a fading variable $\rho\in \mathbb{C}$, a transmission power $P_0$, and a pathloss function as the function of propagation distance. Specifically, the direct pathloss is given by $ \gamma_D = \rho \sqrt{g(d)P_0}$ for distance $d$ while the reflected pathloss is $ \gamma_R = \rho \sqrt{g(d_1)g(d_2)P_0}$, where the latter accounts for the two segments of the reflected path.  
% Here, $\rho\in \mathbb{C}$ represents the small-scale random fading, $g(d)$ denotes the path loss for distance $d$, and $P_0$ is the BS's transmission power. 
Below, we take $g(d) = \beta(d+1)^{-\alpha}$, where the path-loss exponent $\alpha>2$,
and where $\beta = \frac{c}{4\pi f_c}$ is the average power gain for reference distance of $1$m, $f_c$ is the carrier frequency, and $c$ denotes the speed of light. 

The direct links are assumed to experience a Rayleigh fading $\rho_{D}\in \mathbb{C}$. 
As for the reflected links, the fading experienced by one RIS element, denoted by 
$\zeta$, is computed as the product $\rho_{\rm BS-RIS}\cdot \rho_{\rm RIS-UE}$
of two small-scale fading coefficients. This is because each RIS element is modeled as a
lossless phase shifter that can scatter the absorbed energy with a proper phase
shift~\cite{wu2019intelligent, zeng2021reconfigurable}.
Both $\rho_{\rm BS-RIS}$ and $\rho_{\rm RIS-UE}$ are assumed to follow a Rician distribution due to line-of-sight channel conditions, as RISs are usually deployed in favorable locations.
The small-scale fading of a reflected beam, formed by phase alignment of $M$ elements of a RIS panel, can be approximated by a Gaussian distributed random variable. 
This approximation relies on the central limit theorem, since the reflected beam magnitude is the sum of fading amplitude from a large number of RIS elements, given by
$|\rho_R| = \sum_{m=1}^{M}|\zeta^{(m)}| \approx \mathcal{N}(M\mathbb{E}[|\zeta |], M \mathbb{V}[|\zeta|])$,  where $|\zeta^{(m)}|$ is the fading for the $m^{\rm th}$ element.
Consequently, the power of the fading of a reflected beam follows the non-central chi-square distribution, denoted by $ |\rho_R|^2 \sim \chi^2$.

\vspace{-0.1cm}
\subsection{Total Cost of Ownership}
Below, we abstract the TCO by two categories of costs for both BS and RIS, namely, the capital expenditure
(CAPEX) and the operational expenditure (OPEX), denoted by $C_{\rm BS}, C_{\rm RIS}$ and $O_{\rm BS}, O_{\rm RIS}$, respectively.

From the observation that RIS is cost-effective compared to BS, which employs RF-chains and amplifiers~\cite{liu2022path},
we assume that the total CAPEX are such that $C_{\text{BS}} > C_{\text{RIS}}$.
Leveraging the fact that a given RIS is expected to require less power and maintenance~\cite{wang2022reconfigurable}, we also
assume that the total OPEX satisfy $O_{\text{BS}} > O_{\text{RIS}}$. 
The TCO is $\bar{C}_{\text{BS}} = O_{\text{BS}} + C_{\text{BS}}$ for each BS and 
$\bar{C}_{\text{RIS}} = O_{\text{RIS}} + C_{\text{RIS}}$ for each RIS.
Next, we introduce the parameter defining the cost ratio $J=\bar{C}_{\text{BS}}/\bar{C}_{\text{RIS}}$,
which will be used in the techno-economic analysis. 

Note that the $\lambda_{\rm RIS}$ represents the density of RISs within each cluster, which is conditioned on the spatial distribution of BSs.
When evaluating an investment policy for performance enhancement, the cost of deploying additional network elements, whether BSs or RISs, depends on the total number of elements added to an area. 
To achieve a fair comparison of deployment costs, the area-dependent cost, $\Delta \bar{C}_{\rm RIS}$, for increasing the RIS density, $\Delta\lambda_{\rm RIS}$, should be normalized in terms of the total area allocated for RIS deployment, rather than in terms of the cluster area of each cell. 
The normalized cost is given by 
        \vspace{-0.1cm}
\begin{equation}\label{eq:cost_RIS}
	\Delta \bar{C}_{\rm RIS} \propto \bar{C}_{\rm RIS}  \lambda_{\rm BS}\mathcal{A} \Delta \lambda_{\rm RIS},
        \vspace{-0.1cm}
\end{equation}
where $\mathcal{A}$ denotes the area of a cluster ring and $\lambda_{\rm BS}\mathcal{A}$ represents the total area.
Furthermore, when each BS is assumed to associate with $\mathcal{A}  \lambda_{\rm RIS}$ RISs, increasing the BS density $ \Delta \lambda_{\rm BS}$ implies adding the associated RISs, resulting in the cost $ \Delta \bar{C}_{\rm BS}$ due to both BSs and RISs, given by
% Specifically, the area-dependent cost of increasing $\lambda_{\rm BS}$ is proportional to the change in density 
        \vspace{-0.1cm}
\begin{equation}\label{eq:cost_BS}
	\Delta \bar{C}_{\rm BS} = \bar{C}_{\text{BS}} \Delta \lambda_{\rm BS} + \bar{C}_{\text{RIS}} \mathcal{A} \lambda_{\rm RIS}\Delta \lambda_{\rm BS}. 
\end{equation}

\section{Sensitivity Analysis and Investment Strategy}\label{section:analysis}
In this section, we quantify the ergodic rate enhancement of the UE, located either in a throughput-enhanced area or at the center of a coverage hole, as a function of increasing either the BS density or the RIS density. 
In addition, we propose an incremental deployment strategy to assist network operators in making informed decisions about BS and RIS investments.

\subsection{Integral Expression for Mean Ergodic Rate}

Reflected interference from RISs in non-associated cells, denoted by $\sum_{j \in \mathcal{J}_i} | \gamma_{{R}_{i,j}}|^2$ for all $i\neq o$,  can be negligible under certain conditions. 
Specifically, when non-associated reflected beams from other cells are randomly directed and have a low probability of overlapping with the UE of interest, the impact of this interference is minimal, as shown in \cite{sun2023performance}. 
We will empirically verify this condition in the simulation section. 
Assuming this factor has negligible influence on investment decision-making, we will neglect it in our analysis.

A communication link is considered successfully covered when $\text{SINR}\geq T$, where $T$ is the coverage threshold.
By manipulating SINR in Eq.~\eqref{eq:SINR-def} with the coverage threshold $T$, we have
\vspace{-0.2cm}
\begin{equation}\label{eq:manipulated}
    \big|\gamma_{D_o}\big|^2 \geq T\Big(\sum_{i\in \mathcal{I}\setminus o}\big|\gamma_{D_i}\big|^2   + \sigma^2\Big) -  \sum_{j \in \mathcal{J}_o} \big| \gamma_{{R}_{o,j}}\big|^2,  
\vspace{-0.3cm}
\end{equation}
where $I=\sum_{i\in \mathcal{I}\setminus o}\big|\gamma_{D_i}\big|^2$ represents the power of the aggregated interference and $S_R=\sum_{j \in \mathcal{J}_o} \big| \gamma_{{R}_{o,j}}\big|^2$ the power of the reflected signals, respectively. 
Denote the right-hand side of Eq.~\eqref{eq:manipulated} by $\Upsilon\in\mathbb{R}$ as a function of the threshold $T$.
The Laplace transform of $\Upsilon$ can be expressed as the product of the Laplace transform of
the scaled interference power $I$, the noise power $\sigma^2$, and the reflected signals power $-S_R$,
since they are built from independent random point processes and random variables, given
by\footnote{We let $\mathcal{B}(s)$ denote bilateral Laplace transform
and $\mathcal{L}(s)$ denote unilateral Laplace transform.}
$
\mathcal{B}_{\Upsilon}(s) = \mathcal{B}_{I}(sT) \mathcal{B}_{\sigma^2}(sT)  \mathcal{B}_{-S_R}(s).$ 
Specifically, $\mathcal{B}_{I}(sT)$ is obtained from the probability generating
functional (PGFL) of the PPP, given by
        \vspace{-0.1cm}
\begin{equation}\label{eq:laplace_I}
\begin{aligned}
    \mathcal{B}_{I}(sT) =&  e^{ -2\pi \lambda_{\rm BS} \int_{r}^{\infty} x \big(1-\mathcal{L}_{|\rho_{D}|^2}(s T P_0 g(x))\big){\rm d}x}, 
\end{aligned}
\end{equation}
where the Laplace transform of the Rayleigh faded interfering signal is $\mathcal{L}_{|\rho_{D}|^2}\big(s T P_0 g(x)\big) = \frac{1}{1+s T P_0 g(x)}$. In addition, for the Gaussian noise, we have $\mathcal{B}_{\sigma^2}(sT) = e^{-sT\sigma^2}$.
Similarly, the Laplace transform for the reflected signals is built from the PPP
of the RIS cluster associated with the typical BS, given by
        \vspace{-0.1cm}
\begin{equation}\label{eq:laplace_SR}
    \mathcal{B}_{-S_R}(s) = e^{-\lambda_{\rm RIS} \int_{R_{\rm in}}^{R_{\rm out}} \int_0^{2\pi} y\big(1 - \mathcal{L}_{|\rho_R|^2}(-s P_0 G(r, y, \psi)  ) \big) {\rm d}\psi {\rm d}y }.
\end{equation}
Here, the pathloss of the reflected path is $G(r,y,\psi) = g(y)g(\sqrt{r^2+y^2-2ry\cos{\psi}})$,
where $ r= \|\mathbf{x}_o\|, y=\|\mathbf{y}\|$,  represent the scalar distances of the BS and RIS to the UE, respectively, and $\psi \in [0, 2\pi)$ is the angle between the two parts of the reflected link.
In addition, the Laplace transform of the fading power of the reflected beam $|\rho_R|^2 \sim \chi^2$ is
        \vspace{-0.1cm}
\begin{equation}\label{Laplace_transform_rho_R}
    \mathcal{L}_{|\rho_R|^2}(s) = \frac{ \exp\Big(\frac{-s(M\mathbb{E}[|\zeta |])^2}{1+2s M \mathbb{V}[|\zeta|]}\Big)}{(1+2sM \mathbb{V}[|\zeta|])^{1/2}},
        \vspace{-0.1cm}
\end{equation}
where $s > -\frac{1}{2 M \mathbb{V}[|\zeta|]}$ specifies the region of convergence of the Laplace transform expression.

We can now state the integral representation of the coverage probability \cite[Theorem~1]{sun2023performance}:
\begin{theorem}
The coverage probability for a typical UE located at a distance of $r$ to the associated BS is given by
    \vspace{-0.1cm}
    \begin{equation}\label{eq:cov_prob}
        \mathsf{P}_c(T|r) = \mathcal{B}_{\Upsilon^+}\bigg(\frac{1}{  P_0g(r)}\bigg), 
        \vspace{-0.1cm}
    \end{equation}
    where  $\mathcal{B}_{\Upsilon^+}(s)$ is given by the formula
        \vspace{-0.1cm}
    \begin{equation}\label{eq:theorem}
    \mathcal{B}_{\Upsilon^+}(s) = \frac{1}{2\pi \imath}\int_{-\infty}^{\infty} \Big(\mathcal{B}_{\Upsilon}(s-\imath u)- \mathcal{B}_{\Upsilon}(-\imath u) \Big)\frac{{\rm d} u}{u} +  \frac{1+\mathcal{B}_{\Upsilon}(s)}{2}. 
    \end{equation}
Here, $\int_{-\infty}^{\infty}\frac{{\rm d}u}{u}$ is understood in the sense of Cauchy principal-value, that is $ \int_{-\infty}^{\infty}= \lim_{\epsilon \downarrow 0^+}\int_{-\infty}^{-\epsilon}+  \int_{\epsilon}^{\infty}$. 
\end{theorem}
% {\color{blue} The key challenge is differentiating Eq.~\eqref{eq:theorem}, which involves a singularity $u=0$ and an infinite integration domain.}
Based on the Shannon rate formula and the coverage probability, the mean ergodic rate of a typical UE at a distance $r$ from the associated BS is given by an integral:
    \vspace{-0.1cm}
\begin{equation}\label{eq:ergo_rate_at_r}
    \tau(r) =  \int_{0}^{\infty} \mathcal{B}_{\Upsilon^+}\bigg(\frac{1}{P_0g(r)}\bigg)\frac{1}{1+t}\text{d}t. 
    \vspace{-0.1cm}
\end{equation}

\subsection{Sensitivity Analysis for Throughput Enhancement}
To quantify the throughput of RIS-assisted networks, we consider the mean ergodic rate of a UE located at the distance $r\in [R_c, \infty)$ from the BS, given by 
\begin{equation}\label{eq:ergo_rate}
    \tau= \int_{0}^{\infty} \int_{R_c}^{\infty}  \mathcal{B}_{\Upsilon^+}\bigg(\frac{1}{P_0g(r)}\bigg)\frac{1}{1+t} 2\pi \lambda_{\rm BS} r e^{-\pi r^2\lambda_{\rm BS}} \text{d}r \text{d}t, 
\end{equation}
with $\tau(r)$ given in Eq.~\eqref{eq:ergo_rate_at_r}.
In what follows, we assume that RoI is proportional to the mean ergodic rate and we represent the differential RoI by
$\frac{\partial \tau}{\partial \lambda_{\rm BS}}$
and $ \frac{\partial \tau}{\partial \lambda_{\rm RIS}}$.
We will then analyze the investment strategy in Subsection~\ref{subsec_strategy}.

We now derive the first-order derivatives of Eq.~\eqref{eq:ergo_rate} w.r.t. the density of BSs or RISs,
while keeping all other variables fixed, which is the basis for economic analysis.
Differentiating Eq.~\eqref{eq:ergo_rate} requires an interchange of the integration and differentiation operators. 
For this purpose, an appropriate extension of the Leibniz integral rule \cite{bogachev2007measure} is given in Appendix~\ref{sec:appendix}.
Applying the Leibniz integral rule, we get
\begin{equation}\label{eq:partial_lambda_BS}
\begin{aligned}
    \frac{\partial \tau}{\partial\lambda_{\rm BS}} = \int_{0}^{\infty} \int_{R_c}^{\infty}& \bigg( \frac{\partial \mathcal{B}_{\Upsilon^+}(s)}{\partial\lambda_{\rm BS}}  \lambda_{\rm BS}  + \mathcal{B}_{\Upsilon^+}(s) (1   - \pi r^2 \lambda_{\rm BS} ) \bigg)  \\
    &2\pi r e^{-\pi r^2\lambda_{\rm BS}}\frac{1}{t+1}\text{d}r\text{d}t.
\end{aligned}
\vspace{-0.1cm}
\end{equation}
Here, we determine $\frac{\partial \mathcal{B}_{\Upsilon^+}(s)}{\partial\lambda_{\rm BS}}$  through another application of the Leibniz integral rule to Eq.~\eqref{eq:theorem}, which also involves interchanging of integration and differentiation. 
While the Leibniz integral rule requires the function under the integral sign to be a continuous function, the singular integral in Eq.~\eqref{eq:theorem} has a singularity at zero and at limit points at infinity. 
The conditions of this interchange are verified in Appendix~\ref{sec:appendix}.
Subsequently, we have
\vspace{-0.1cm}
\begin{equation}\label{eq:first_order_derivative_BS}
\begin{aligned}
        &\frac{\partial \mathcal{B}_{\Upsilon^+}(s)}{\partial\lambda_{\rm BS}} = \\
        &\frac{1}{2\pi \imath}\int_{-\infty}^{\infty} \bigg(\frac{\partial\mathcal{B}_{\Upsilon}(s-\imath u)}{\partial\lambda_{\rm BS}}- \frac{\partial\mathcal{B}_{\Upsilon}(-\imath u)}{\partial\lambda_{\rm BS}} \bigg)\frac{{\rm d} u}{u}  
        + \frac{1}{2}\frac{\partial\mathcal{B}_{\Upsilon}(s)}{\partial\lambda_{\rm BS}},
\end{aligned}
\end{equation}
where the expression for the derivative of the Laplace transform $\mathcal{B}_{\Upsilon}(s)$ w.r.t. the BS density is given by
\begin{equation}\label{eq:partial_BS}
\begin{aligned}
    &\frac{\partial \mathcal{B}_{\Upsilon}(s) }{\partial\lambda_{\rm BS}} = -\mathcal{B}_{\Upsilon}(s) \cdot 2\pi \int_{r}^{\infty} x \Big(1-\mathcal{L}_{|\rho_{D}|^2}\big(s T P_0 g(x)\big)\Big){\rm d}x. 
\end{aligned}
\end{equation}

Similar to Eq.~\eqref{eq:partial_lambda_BS}, we can derive the following expression for the partial derivative of $\tau$ w.r.t. $\lambda_{\rm RIS}$, given by
\begin{equation}\label{eq:partial_lambda_RIS}
\begin{aligned}
    \frac{\partial \tau}{\partial\lambda_{\rm RIS}} = \int_{0}^{\infty} \int_{R_c}^{\infty} \bigg( \frac{\partial \mathcal{B}_{\Upsilon^+}(s)}{\partial\lambda_{\rm RIS}}   \bigg)  2\pi r\lambda_{\rm BS} e^{-\pi r^2\lambda_{\rm BS}}\frac{1}{t+1}\text{d}r\text{d}t,
\end{aligned}
\end{equation}
where $\frac{\partial \mathcal{B}_{\Upsilon^+}(s)}{\partial\lambda_{\rm RIS}}$ can be expressed following the steps in Eq.~\eqref{eq:first_order_derivative_BS} and Eq.~\eqref{eq:partial_BS}.

\subsection{Sensitivity Analysis for Coverage Hole Mitigation}
To analyze the impact of RIS on coverage hole mitigation, we are interested in the performance in terms of the ergodic rate of the UE located in the center of the coverage hole. 
In this case, assuming a fixed BS-UE distance eliminates the need for integration over the distance $r$. Consequently, $\tau$ is given by Eq.~\eqref{eq:ergo_rate_at_r} instead of Eq.~\eqref{eq:ergo_rate}.
Moreover, the derivatives obtained in the previous subsection require further modifications.

Recall that we introduce a penalty factor of $K$ to account for the impact of blocked direct links. 
This penalty factor influences the argument of the Laplace transform to express the coverage probability in Eq.~\eqref{eq:cov_prob}, which is then modified to $\frac{1}{KP_0g(r_{\rm H})}$.
Furthermore, changing the BS density also varies the average distance $r_{\rm H}$ between coverage hole centers and their cell centers, as defined in Eq.~\eqref{eq:cell_center}.
% Note that the rate of change of the ergodic rate is represented by the derivative of the Laplace transform of $\Upsilon$, i.e., $\frac{\partial \mathcal{B}_{\Upsilon}(s)}{\partial \lambda_{\rm BS}}$. 
In the following, we need to consider two factors due to the dependence of $r_{\rm H}$ on $\lambda_{\rm BS}$. First, the non-associated BSs within a distance of $[r_{\rm H},\infty)$ that introduce interference are closer to the UE for larger values of $\lambda_{\rm BS}$.  Second, the argument $s=\frac{1}{KP_0g(r_{\rm H})}$ is a function of $r_{\rm H}$ and consequently depends on  $\lambda_{\rm BS}$. Let $f_s(r_{\rm H}) = \frac{1}{KP_0g(r_{\rm H})}$ be the function describing this dependence in the following. The changing rates of these two factors are given by 
\vspace{-0.1cm}
\begin{equation}\label{eq:additional_derivative}
	\frac{\partial r_{\rm H}}{\partial \lambda_{\rm BS}} = -\frac{C_{\rm H}}{2\sqrt{\lambda_{\rm BS}^3}}, ~\frac{\partial f_s(r_{\rm H})}{\partial \lambda_{\rm BS}} = -\frac{1}{KP_0g^2(r_{\rm H})}\frac{\partial g(r_{\rm H})}{\partial r_{\rm H}}\frac{\partial r_{\rm H}}{\partial \lambda_{\rm BS}}. 
\end{equation}
Recall that $\mathcal{B}_{\Upsilon}(s) = \mathcal{B}_{I}(sT)\mathcal{B}_{-S_R}(s)\mathcal{B}_{\sigma^2}(sT)$, we have 
\vspace{-0.1cm}
\begin{equation}\label{eq:products}
\begin{aligned}
\frac{\partial \mathcal{B}_{\Upsilon}(s)}{\partial \lambda_{\rm BS}} =& \frac{\partial \mathcal{B}_{I}(sT)}{\partial \lambda_{\rm BS}}   \mathcal{B}_{-S_R}(s) \mathcal{B}_{\sigma^2}(sT) + \frac{\partial \mathcal{B}_{-S_R}(s)}{\partial \lambda_{\rm BS}} \mathcal{B}_{I}(sT) \\&  \mathcal{B}_{\sigma^2}(sT) + \frac{\partial \mathcal{B}_{\sigma^2}(sT)}{\partial \lambda_{\rm BS}} \mathcal{B}_{I}(sT) \mathcal{B}_{-S_R}(s).  
\end{aligned}
\vspace{-0.13cm}
\end{equation}
We will discuss the derivatives in Eq.~\eqref{eq:products} separately. First, for the derivative related to the interference $I$, we have 
\vspace{-0.1cm}
\begin{equation}\label{eq:derivaive_I}
	\frac{\partial \mathcal{B}_{I}(sT)}{\partial \lambda_{\rm BS}} = \mathcal{B}_{I}(sT) \Big(\mathcal{D}_{\rm BS}(sT) + \lambda_{\rm BS} \frac{\partial \mathcal{D}_{\rm BS}(sT) }{\partial \lambda_{\rm BS}}\Big), 
\end{equation}
where $\mathcal{D}_{\rm BS}(sT)=-2\pi \int_{r_{\rm H}}^{\infty} \frac{xsTP_0g(x)}{1+sTP_0g(x)}{\rm d}x$ from Eq.~\eqref{eq:laplace_I}. 
In Eq.~\eqref{eq:derivaive_I}, the expressions of $\frac{\partial \mathcal{D}_{\rm BS}(sT) }{\partial \lambda_{\rm BS}}$ are different when differentiating $\mathcal{B}_{\Upsilon}(s-\imath u)$ and $\mathcal{B}_{\Upsilon}(-\imath u)$ in Eq.~\eqref{eq:theorem}, respectively, For the former, both $r_{\rm H}$ and $s$ are dependent on $\lambda_{\rm BS}$.  
To differentiate $\mathcal{D}_{\rm BS}(sT) $ w.r.t. $\lambda_{\rm BS}$, another application of the Leibniz rule is needed. 
The condition for this application is discussed in Appendix~\ref{app:leibniz_coverage_hole}.
As a result, $\frac{\partial \mathcal{D}_{\rm BS}(sT) }{\partial \lambda_{\rm BS}}$ is given by
\begin{equation}\label{partial_D_I}
\begin{aligned}
\frac{\partial \mathcal{D}_{\rm BS}(sT) }{\partial \lambda_{\rm BS}} = &2\pi \frac{r_{\rm H}sTP_0g(r_{\rm H})}{1+sTP_0g(r_{\rm H})}\frac{\partial r_{\rm H}}{\partial \lambda_{\rm BS}} \\& -\int_{r_{\rm H}}^{\infty} \frac{2\pi xTP_0g(x)}{\big(1+sTP_0g(x)\big)^2} \frac{\partial f_s(r_{\rm H})}{\partial \lambda_{\rm BS}}{\rm d}x,  
\end{aligned}
\end{equation}
where $\frac{\partial r_{\rm H}}{\partial \lambda_{\rm BS}}$ and $\frac{\partial f_s(r_{\rm H})}{\partial \lambda_{\rm BS}}$ are given in Eq.~\eqref{eq:additional_derivative}. 
For the latter, $\mathcal{B}_{\Upsilon}(-\imath u)$, where the argument $ -\imath u$ does not depend on $r_{\rm H}$, the derivative  $\frac{\partial \mathcal{D}_{\rm BS}(sT) }{\partial \lambda_{\rm BS}}$ only contains the first part of Eq.~\eqref{partial_D_I} since $\frac{\partial f_s(r_{\rm H})}{\partial \lambda_{\rm BS}}$ in the second term is equal to zero.

Next, we can express the derivative of $\mathcal{B}_{-S_R}(s)$ as
\begin{equation}\label{eq:Derivative_S_R_BS}
\begin{aligned}
\vspace{-0.3cm}
	&\frac{\partial \mathcal{B}_{-S_R}(s)}{\partial \lambda_{\rm BS}} = \mathcal{B}_{-S_R}(s)\lambda_{\rm RIS} \frac{\partial \mathcal{D}_{\rm RIS}(s)}{\partial \lambda_{\rm BS}},
\end{aligned}
\vspace{-0.14cm}
\end{equation} 
where we use $\mathcal{D}_{\rm RIS}(s)$ to abbreviate $$-\int_{0}^{2\pi}\int_{R_{\rm in}}^{R_{\rm out}} y\Big(1- \mathcal{L}_{|\rho_R|^2}\big(-s P_0 G(r_{\rm H}, y, \psi)\big)\Big) {\rm d}y {\rm d}\psi.$$
Note that $\frac{\partial \mathcal{D}_{\rm RIS}(s)}{\partial \lambda_{\rm BS}}$ necessitates the application of the Leibniz rule to interchange differentiation and integration, given by
\begin{equation}
\begin{aligned}
&-\frac{\partial}{\partial \lambda_{\rm BS}} \int_{0}^{2\pi}\int_{R_{\rm in}}^{R_{\rm out}} y\Big(1- \mathcal{L}_{|\rho_R|^2}\big(-s P_0 G(r_{\rm H}, y, \psi)\big)\Big) {\rm d}y {\rm d}\psi  \\
    &=\int_{0}^{2\pi}\int_{R_{\rm in}}^{R_{\rm out}} \bigg(y \frac{\partial\mathcal{L}_{|\rho_R|^2}\big(-s P_0 G(r_{\rm H}, y, \psi)\big)}{\partial \lambda_{\rm BS}}\bigg) {\rm d}y {\rm d}\psi , 
\end{aligned}
\end{equation}
where the condition is discussed in Appedix~\ref{app:leibniz_coverage_hole}. 
When differentiating $\mathcal{B}_{\Upsilon}(s-\imath u)$ and $\mathcal{B}_{\Upsilon}(-\imath u)$, 
note that in the former, both $G(r_{\rm H}, y, \psi)$ and $s$ are dependent on $r_{\rm H}$, whereas in the latter only $G(r_{\rm H}, y, \psi)$ depends on $r_{\rm H}$. It is important to express the derivatives accounting for these factors. 
For $\mathcal{B}_{\Upsilon}(s-\imath u)$, the derivative of Eq.~\eqref{Laplace_transform_rho_R} is 
\vspace{-0.1cm}
\begin{equation}\label{eq:derivative_S_R_additional}
\begin{aligned}
&\frac{\partial\mathcal{L}_{|\rho_R|^2}\big(-s P_0 G(r_{\rm H}, y, \psi)\big)}{\partial \lambda_{\rm BS}} =\\&
	\bigg(\frac{ M(\mathbb{E}[|\zeta|])^2/\mathbb{V}[|\zeta|]-2 s P_0M\mathbb{V}[|\zeta|]G(r_{\rm H},y,\psi)+1 }{(1-2 s P_0G(r_{\rm H},y,\psi)M\mathbb{V}[|\zeta|])^{5/2}}\bigg)
	\\
	& P_0M\mathbb{V}[|\zeta|]\exp{\bigg(\frac{ s P_0G(r_{\rm H},y,\psi)(M\mathbb{E}[|\zeta|])^2 }{1-2s P_0G(r,y,\psi)) M\mathbb{V}[|\zeta|] }\bigg)}\\&\bigg(s\frac{\partial	G(r_{\rm H},y,\psi)}{\partial \lambda_{\rm BS}}+G(r_{\rm H},y,\psi)\frac{\partial	f_s(r_{\rm H})}{\partial \lambda_{\rm BS}} \bigg),
\end{aligned}
\vspace{-0.1cm}
\end{equation}
where $\frac{\partial G(r_{\rm H},y,\psi)}{\partial \lambda_{\rm BS}} = \frac{\partial	G(r,y,\psi)}{\partial r_{\rm H}} \cdot  \frac{\partial	r_{\rm H}}{\partial \lambda_{\rm BS}}$, with $\frac{\partial r_{\rm H}}{\partial \lambda_{\rm BS}}$ and $\frac{\partial f_s(r_{\rm H})}{\partial \lambda_{\rm BS}}$ given in Eq.~\eqref{eq:additional_derivative}.  
For differentiating $\mathcal{B}_{\Upsilon}(-\imath u)$, where the argument $- \imath u$ does not depend on $r_{\rm H}$, the factor $G(r_{\rm H},y,\psi)\frac{\partial	f_s(r_{\rm H})}{\partial \lambda_{\rm BS}}$ in Eq.~\eqref{eq:derivative_S_R_additional} is zero. 
% this derivative is given by
% \begin{equation}\label{eq:derivative_S_R_additional_im}
% \begin{aligned}
% &\frac{\partial\mathcal{L}_{|\rho_R|^2}\big(-s P_0 G(r_{\rm H}, y, \psi)\big)}{\partial \lambda_{\rm BS}} = \\&\frac{s P_0 M_0\sigma_{\chi}^2 (M_0\mu_{\chi}^2/\sigma^2_{\chi}-2 s M_0\sigma_{\chi}^2 P_0G(r_{\rm H},y,\psi)+1)}{(1-2 s P_0G(r_{\rm H}, y,\psi)M_0\sigma_{\chi}^2)^{5/2}}  \\
% 	& \exp{\bigg(\frac{ s M_0^2\mu_{\chi}^2 P_0G(r_{\rm H}, y,\psi)}{1-2 M_0\sigma_{\chi}^2 s P_0G(r_{\rm H},y,\psi))}\bigg)} \frac{\partial	G(r_{\rm H}, y,\psi)}{\partial \lambda_{\rm BS}}, 
% \end{aligned}
% \end{equation}
% with $\frac{\partial G(r_{\rm H},y,\psi)}{\partial \lambda_{\rm BS}}$ defined above for Eq.~\eqref{eq:derivative_S_R_additional}.

Last, we express the derivative of $\mathcal{B}_{\sigma^2}(sT)$ as 
\begin{equation}
	\frac{\partial \mathcal{B}_{\sigma^2}(sT)}{\partial \lambda_{\rm BS}} = -T\sigma^2e^{-sT\sigma^2}\frac{\partial f_s(r_{\rm H})}{\partial \lambda_{\rm BS}},
\vspace{-0.2cm}
\end{equation}
where $\frac{\partial f_s(r_{\rm H})}{\partial \lambda_{\rm BS}}$ is given in Eq.~\eqref{eq:additional_derivative}. 

The above modifications are necessary for the derivative w.r.t. $\lambda_{\rm BS}$. However, for the derivative w.r.t. $\lambda_{\rm RIS}$, no modification is needed. This is because adding RISs does not change the underlying spatial deployment of the BSs, and the expression in Eq.~\eqref{eq:partial_lambda_RIS}, modified by removing the integration over the distribution of $r$, gives the derivative for this scenario.

\vspace{-0.2cm}
\subsection{Incremental Investment Strategy}\label{subsec_strategy}
We then show how the above derivatives allow one to develop a strategic approach for incremental investment in deploying RIS-enhanced cellular systems. 
Recall that the RoI is assumed to be proportional to the ergodic rate, while the TCO is proportional to the number of deployed BSs and RISs.    
In each investment round, we use the partial derivatives $\frac{\partial \tau}{\partial\lambda_{\rm BS}}$ and $\frac{\partial \tau}{\partial\lambda_{\rm RIS}}$ to quantify the expected increase in ergodic rate, denoted by $E$, as a result of an increase in the density of either BS or RIS due to an investment. 
Therefore, we assume $E_{\rm BS} = \frac{\partial \tau}{\partial\lambda_{\rm BS}}$. However, note that the number of RIS in a given area depends not only on the RIS density but also on the number of BS and the size of the RIS cluster ring. As a result, we have $E_{\rm RIS} = \lambda_{\rm BS}\mathcal{A} \frac{\partial \tau}{\partial\lambda_{\rm RIS}}$, where $\mathcal{A}$ is the area of the cluster ring.

Here, the decision-making process is determined by comparing the expected gain ratio to the cost ratio to identify the more cost-effective investment. 
Recall that we denote the cost ratio $J = \frac{\bar{C}_{\rm BS}}{\bar{C}_{\rm RIS}}$. The incremental costs, $ \Delta \bar{C}_{\rm BS}$ and $ \Delta \bar{C}_{\rm RIS}$, are defined in  Eq.~\eqref{eq:cost_RIS} and Eq.~\eqref{eq:cost_BS}, respectively.
The policy is then given in Algorithm~\ref{alg:strategy}.
% Note that when $\lambda_{\rm RIS}\neq 0$, deploying a new BS implies additional RISs.
% In other words, assuming a fixed investment budget in each round, the cost of deploying a BS with RISs is modified into
% \begin{equation}
%   \Delta C_{\rm BS} = \bar{C}_{\rm BS} + \lambda_{\rm RIS}\mathcal{A} \bar{C}_{\rm RIS} = (1+\lambda_{\rm RIS}\mathcal{A}/J)\bar{C}_{\rm BS}. 
% \end{equation}
\begin{algorithm}[ht]
\caption{Policy for investment.}\label{alg:strategy}
\begin{algorithmic}
\State Given the cost ratio $J$.\Comment{Total Cost of Ownership.}
\While{investing}
\State Get $\lambda_{\rm BS}, \lambda_{\rm RIS}$.
\State Sensitivity to $\lambda_{\rm BS}$: $E_{\rm BS}=\frac{\partial \tau}{\partial\lambda_{\rm BS}}.$
\State Sensitivity to $\lambda_{\rm RIS}$:  $E_{\rm RIS}=\lambda_{\rm BS}\mathcal{A} \frac{\partial \tau}{\partial\lambda_{\rm RIS}}.$
\If {$\frac{E_{\rm BS}}{E_{\rm RIS}} \geq J(1+\lambda_{\rm RIS}\mathcal{A}/J)$ }
\State Allocate the investment to increase the BS density
\Else
\State Allocate the investment to increase the RIS density
\EndIf
\EndWhile
\end{algorithmic}
\end{algorithm}

% \vspace{-0.2cm}
\section{Numerical Results and Decision Analysis}\label{sec:simulation}

The system parameters for the numerical analysis are as follows unless otherwise specified. 
Throughout the simulation, we consider a BS density ranging from 3 to 40 BS/km$^{2}$, resulting in average cell radii ranging from approximately 80 to 300 meters. 
The geometry of RIS clusters is fixed. It has an inner radius of 20 meters and an outer radius of 30 meters. 
Plot labels indicate the number of RIS elements of each RIS panel, representing varying beamforming capabilities.
We consider that each BS transmits at $P_0 = 30$dBm, and the noise power is set as -100dBm. 

\vspace{-0.2cm}
\subsection{Performance of RISs for Throughput Enhancement }
\begin{figure}[htp!]
    \centering
    \includegraphics[width=0.85\linewidth]{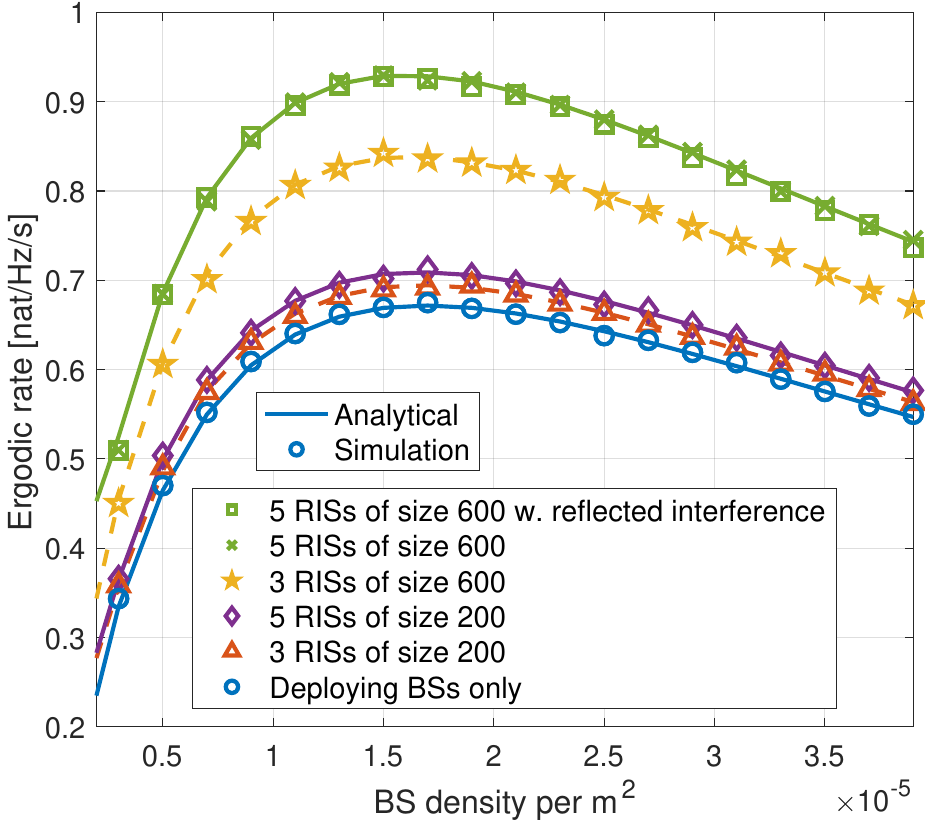}
    \caption{Mean ergodic rate of a typical UE as a function of the density of BSs for different RIS densities.}
    \label{rate_overall}
    \vspace{-0.3cm}
\end{figure}

\begin{figure}
    \centering
    \includegraphics[width=0.82\linewidth]{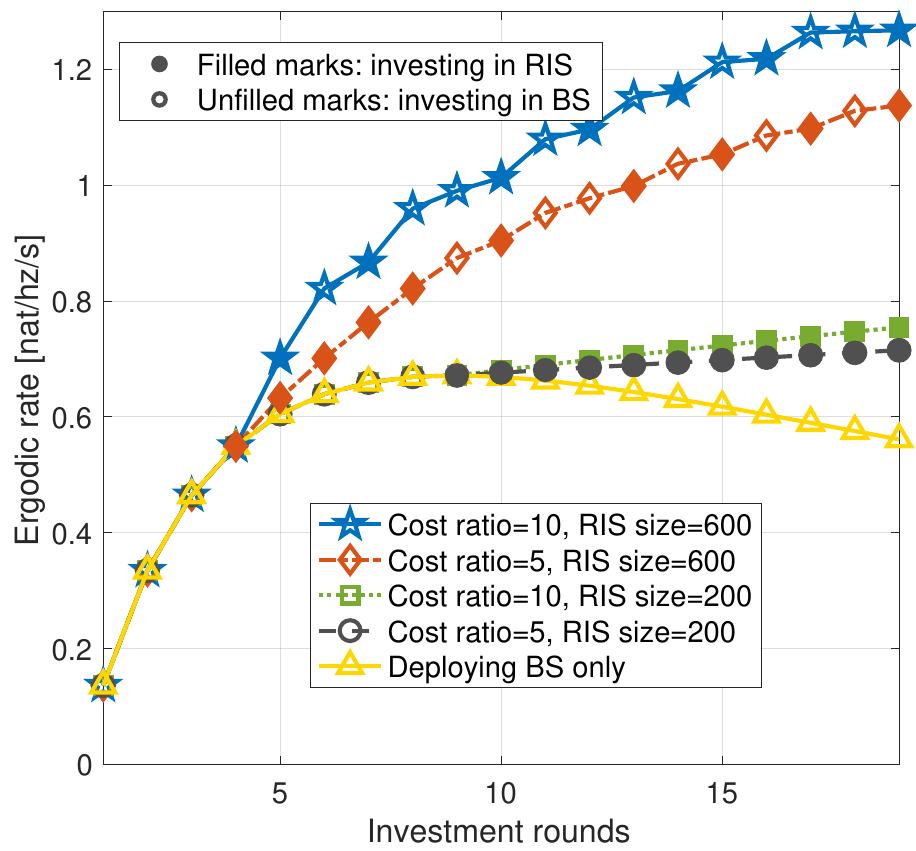}
    \caption{Ergodic rate evolution based on investment in either BSs or RISs using Algorithm~\ref{alg:strategy}}
    \label{fig:decision_map}
    \vspace{-0.4cm}
\end{figure}

\begin{figure}
    \centering
    \includegraphics[width=0.88\linewidth]{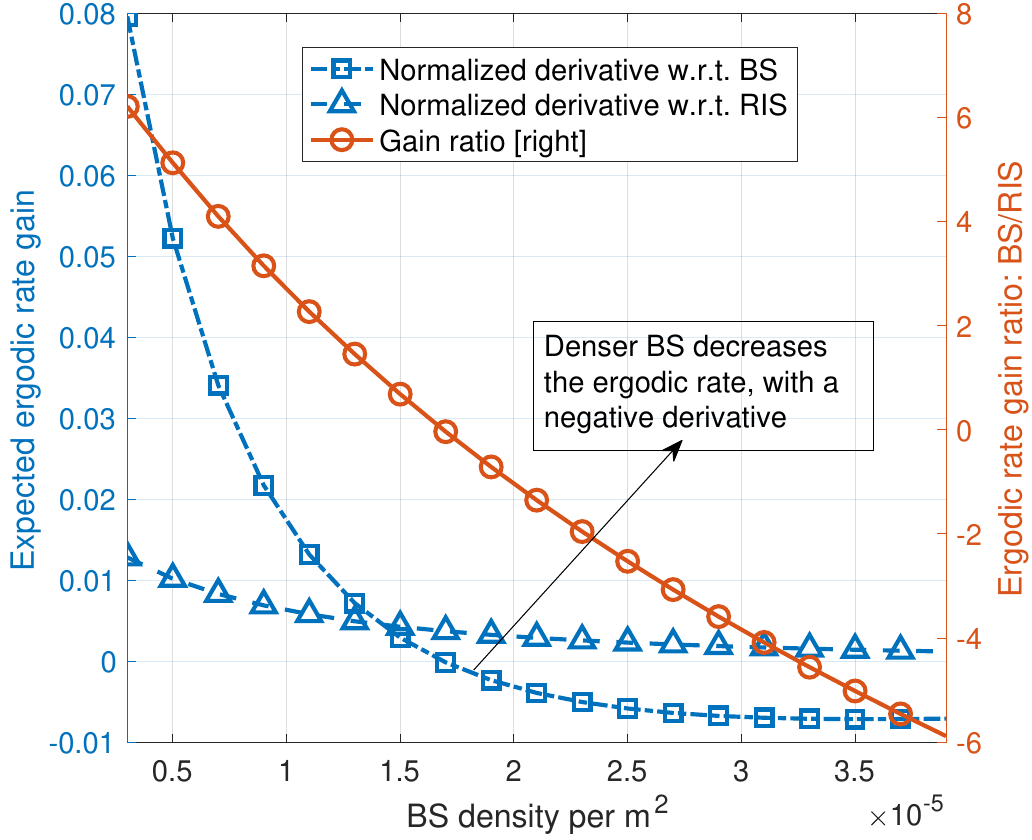}
    \caption{Should one deploy more BSs before deploying RISs?}
    \label{first_order_improve}
    \vspace{-0.6cm}
\end{figure}

Fig.~\ref{rate_overall} depicts the ergodic rate of a typical UE, where the guard zone $R_c$ is $50$ meters\footnote{This guard distance ensures that the UE is not in the near field of either BS or RIS.}, as a function of the BS density\footnote{The BS density under consideration spans a range where
the inter-BS distance is significantly greater than the cluster radius, ensuring that the RISs are in
proximity to their associated BSs. For example, the average inter-BS distance for
$\lambda_{\rm BS}=40/{\rm km}^2$ is approximately $80$ meters, which exceeds twice the outer radius (30 meters).}. 
It also illustrates the impact of varying RIS densities of each cluster by considering 3 or 5 RISs per cluster.
From Fig.~\ref{rate_overall}, we observe that the mean ergodic rate is primarily constrained by the Gaussian
noise when the BS density is low, and by the aggregated interference when the BS density becomes high. 
Note that the densification of BSs can be detrimental (in terms of spectral efficiency) above a certain BS density. 
This observation is consistent with the findings reported in~\cite{alammouri2017sinr}.  
Additionally, a higher BS density reduces the likelihood of a randomly located UE being outside the constant guard distance, reducing the probability of RIS-assisted communication. This also explains the decrease of the ergodic rate for a very high BS density. 
Furthermore, we simulate the scenario where all RISs reflect interference to other cells with a probability of beam overlapping $1\%$ (for a beam width of 3.6 degrees\footnote{ The half-power beamwidth for a 64-element RIS is estimated as 1-2 degrees in \cite{han2020half}.}), depicted by green square marks in Fig.~\ref{rate_overall}.  
By observing the overlap between the cross marks and square marks, we confirm that reflected interference is negligible when the overlap probability is low.
Comparing the curves, we observe that the size of RIS is the most significant factor in improving the ergodic rate, with 600-element RIS significantly outperforming the 200-element one.

In Fig.~\ref{fig:decision_map}, we discuss the successive deployment strategy for the throughput enhancement scenario based on Algorithm~\ref{alg:strategy}.
The initial scenario assumes a deployment density of $1$ BSs per ${\rm km}^2$ without RISs.
In each round, the operator has the investment budget that can deploy either $2$ BSs per ${\rm km}^2$ or equivalently $\bar{C}_{\rm BS}/\bar{C}_{\rm RIS}$ RISs. 
The investment is then directed toward deploying the type of node that yields a larger incremental gain.
In particular, we assume cost ratios $\bar{C}_{\rm BS}/\bar{C}_{\rm RIS}$ to be 5 or 10, and consider the size of RIS to be either 200 or 600 elements. 
In Fig.~\ref{fig:decision_map}, filled marks indicate RIS investment, while unfilled marks denote BS investment.
Fig.~\ref{fig:decision_map} demonstrates that a higher cost ratio not only results in an earlier decision to deploy RISs but also leads to a higher overall performance gain due to the ability to deploy more RISs with the same investment.
Specifically, it is observed that the ergodic rate grows slowly when deploying small RISs (200 elements). However, RISs can still provide a solution for mitigating negative impacts in dense BS deployments. 
On the contrary, when RISs are larger (600 elements), they are capable of providing stronger reflected signals. The ergodic rate evolution suggests a preference for hybrid BS-RIS deployments.

Fig.~\ref{first_order_improve} compares the expected ergodic rate gain from adding one BS or one RIS (600 elements) per $\text{km}^2$, based on first-order derivatives of the mean ergodic rate w.r.t. the BS or RIS density\footnote{Note that the RIS density derivative is defined per cluster. To enable a fair comparison between the impact of adding one RIS or one BS per km${}^2$, we scaled the derivatives accordingly.}. 
This figure shows in the first place where the expected ergodic rate gain is positive and where it is negative.
Note that both blue curves have a downward trend, suggesting a diminishing return on investment.
A higher BS density means not only a stronger signal strength but also a higher interference power level. 
Indeed, an excessive investment allocation on BSs ultimately leads to negative returns~\cite{alammouri2017sinr}. 
In scenarios with high BS density, implying excessive interference, the investment in RIS emerges as the only solution, since the ergodic rate would degrade when deploying extra BSs, aligning with the investment decision in Fig.~\ref{fig:decision_map}. 
Furthermore, the red curve represents the ratio of ergodic rate gains between deploying BSs or RISs. 
This characterizes the technological gain ratio between investing in an additional BS or RIS. 
An informed economic decision is made by comparing this ratio with the cost ratio coefficient. 
For example, in a scenario with BS density of $10^{-5}$/m${}^2$, it is more cost-effective to invest in RISs when the cost ratio exceeds 3, implying that deploying three RISs offers a better performance improvement than a BS with a lower TCO. 

\subsection{Evaluating the Impact of Coverage Hole Mitigation}
% , this distance is approximately 150 meters. 
Recall that $C_{\rm H}$ scales this distance in Eq.~\eqref{eq:cell_center}.
To ensure the RIS ring is deployed mostly within the cell, we select $C_{\rm H}$ such that the average distance from the coverage hole to the BS is 80 meters\footnote{For two-dimensional PPP-modeled cellular networks, the average distance between any point and its nearest BS is $1/(2\sqrt{\lambda_{\rm BS}})$, approximately 150 meters for 
$\lambda_{\rm BS}=10^{-5}/{\rm m}^2$.} for middle-range coverage holes when $\lambda_{\rm BS}=10^{-5}/{\rm m}^2$.
We will focus on examining how the direct link blockage penalty influences the deployment strategy. 
Recall that the expected ergodic rate gain ratio (Fig.~\ref{first_order_improve}) is introduced to assess technological preference compared with economic cost. 
We plot this ratio for the coverage hole scenario to provide insights into deployment decisions, which in turn lead to an investment policy.

\begin{figure}
    \centering
    \includegraphics[width=0.85\linewidth]{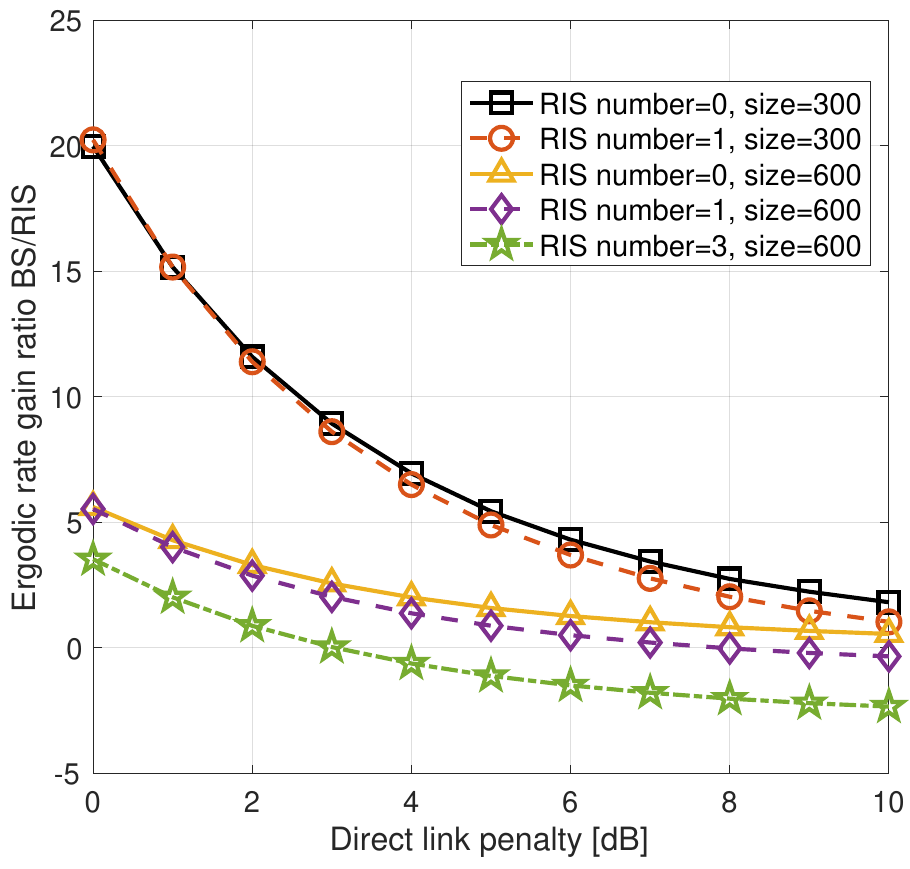}
    \caption{Expected ergodic rate gain ratio between adding new BS or RIS as a function of direct link penalty}
    \label{fig:gain_coverage_penalty}
    \vspace{-0.6cm}
\end{figure}

Fig.~\ref{fig:gain_coverage_penalty} shows the expected gain ratio as a function of the direct link penalty $K$, which takes values from $0$ to $10$ dB. 
We consider the parameters including the size of RISs and the number of existing RISs in each cell. 
% We also investigate the impact of the distance of an average coverage hole to the associated BS.   
In general, these curves in Fig.~\ref{fig:gain_coverage_penalty} decrease, suggesting that, when the direct link to the UE in the coverage hole is severely degraded, the preference for the deployment of BS over RIS decreases. 
Fig.~\ref{fig:gain_coverage_penalty} shows that a moderate cost ratio (e.g., 5) could favor deploying large RISs (600 elements), but may not incentivize deploying small RISs (300 elements) unless the direct link is severely penalized. 
As observed in the curves marked with circles, diamonds, and stars, when BSs are already equipped with a certain number of RISs per cluster, increasing the BS density becomes less favorable since it implies the cost of adding new RISs in the area. 
% This can also be shown by comparing the curves with squares and circles. 
In the specific case where each BS is configured with an average of three RISs, as shown by the curve marked by green stars, deploying a new BS equipped with three RISs may be less advantageous than allocating these three RISs to existing BSs, given that the expected gain ratio is mostly negative.
% RIS size significantly influences deployment preferences, as seen in the two clusters of curves, as a larger RIS can beamform a stronger reflected beam. 
% Distance from the coverage hole to the BS has a less significant impact than other factors, as shown by the triangle and cross curves. This allows one to generalize the aforementioned conclusions to coverage hole scenarios of different distances to their associated BS.

% Before introducing RISs, the current scenario, i.e., the BS density will influence the investment preference, as shown in Fig.~\ref{}. 
\begin{figure}
    \centering
    \includegraphics[width=0.85\linewidth]{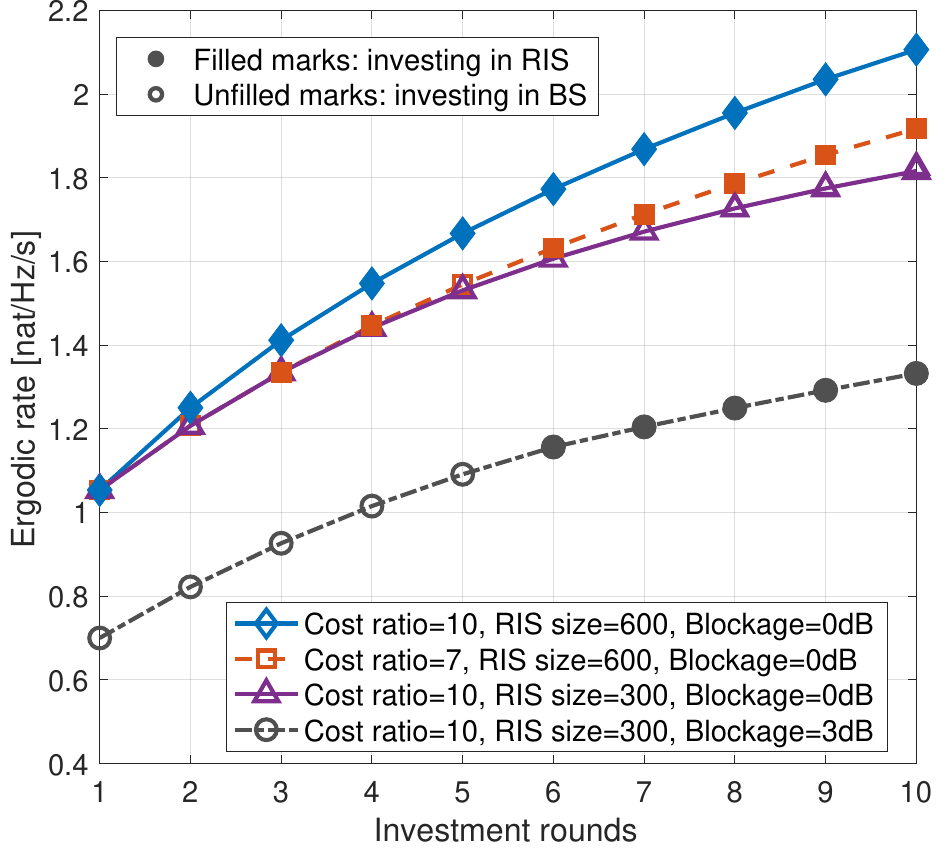}
    \caption{Ergodic rate evolution for coverage hole mitigation scenario based on investment in either BSs or RISs}
    \label{fig:policy_coverage_hole}
    \vspace{-0.6cm}
\end{figure}
The successive investment plan is plotted in Fig.~\ref{fig:policy_coverage_hole}. 
The scenario starts with a BS density of $5/{\rm km}^2$ without any RISs. In each investment round, one can either deploy a BS or several RISs according to the cost ratio $\bar{C}_{\rm BS}/\bar{C}_{\rm RIS}$. 
As in Fig.~\ref{fig:decision_map}, filled marks indicate RIS investment and unfilled marks denote BS investment in each round.
Without considering blockage, the performance improvement due to RIS is compared with the purple curve marked with triangles,  which solely invests in BSs. This is because deploying 300 elements RIS cannot outperform adding BSs with the same investment.
Different branches emerge for varying cost ratios and investment decisions when the RIS is equipped with a significant number of elements for beamforming (600 elements).
When the cost ratio is 7, as depicted by the curve with square marks, the optimal strategy involves alternating investments between BSs and RISs, resulting in a modest performance enhancement.
When RISs are both large (600 elements) and cost-efficient ($\bar{C}_{\rm BS}/\bar{C}_{\rm RIS}=10$), deploying RIS is always preferable, as shown by the blue curve marked with diamonds.
% When the cost ratio is high (with a value of 10), implying that RISs are both affordable and powerful, investments in RIS consistently provide the optimal ergodic rate.
Moreover,  Fig.~\ref{fig:policy_coverage_hole} shows that the ergodic rate is significantly influenced by the blockage. 
The black curve, representing a 3dB penalty, is distinctly separated from the other curves that have no penalty. 
This scenario also turns into investing in RIS when the BS density exceeds a certain threshold.
\vspace{-0.1cm}

\section{Conclusion}\label{sec:conclusion}
% \vspace{-0.1cm}

A well-designed RIS deployment strategy is key to successful cellular network densification. 
We proposed a first techno-economic analysis for optimizing network investments by modeling the RIS-assisted cellular network using stochastic geometry and considering the relative costs of deploying BSs and RISs.
A sensitivity analysis is conducted to assess the impact of parameter changes, including the density
of BSs and RIS, and their unit cost, on the system's ergodic rate for both throughput enhancement and coverage hole mitigation scenarios. 
Analytical expressions are derived to devise an optimal deployment strategy based on techno-economic considerations. 
Numerical results illustrate how investment decisions, determined by cellular network conditions, impact the evolution of ergodic rates for both scenarios.
This work provides insights into techno-economic factors for operators to maximize RoI when deploying RIS-assisted cellular networks.
\section*{Acknowledgement}
% \vspace{-0.1cm}
The work of F. Baccelli and G. Sun was supported by the European
Research Council project titled NEMO under grant ERC 788851.
The work of F. Baccelli was also supported by the French National Agency
for Research
project titled France 2030 PEPR r\'{e}seaux du Futur under grant
ANR-22-PEFT-0010.

\bibliographystyle{ieeetr}
\vspace{-0.1cm}
\bibliography{reference}

\begin{thebibliography}{10}

\bibitem{xu2020understanding}
D.~Xu, A.~Zhou, X.~Zhang, G.~Wang, X.~Liu, C.~An, Y.~Shi, L.~Liu, and H.~Ma, ``Understanding {O}perational 5{G}: {A} {F}irst {M}easurement {S}tudy on its {C}overage, {P}erformance and {E}nergy {C}onsumption,'' in {\em Proceedings of the Annual conference of the ACM Special Interest Group on Data Communication on the {A}pplications, Technologies, Architectures, and Protocols for Computer Communication}, pp.~479--494, 2020.

\bibitem{liu2022path}
R.~Liu, Q.~Wu, M.~Di~Renzo, and Y.~Yuan, ``A {P}ath to {S}mart {R}adio {E}nvironments: {A}n {I}ndustrial {V}iewpoint on {R}econfigurable {I}ntelligent {S}urfaces,'' {\em IEEE Wireless Communications}, vol.~29, no.~1, pp.~202--208, 2022.

\bibitem{wang2021joint}
J.~Wang, Y.-C. Liang, J.~Joung, X.~Yuan, and X.~Wang, ``Joint {B}eamforming and {R}econfigurable {I}ntelligent {S}urface {D}esign for {T}wo-{W}ay {R}elay {N}etworks,'' {\em IEEE Transactions on Communications}, vol.~69, no.~8, pp.~5620--5633, 2021.

\bibitem{lang2009business}
E.~Lang, S.~Redana, and B.~Raaf, ``Business {I}mpact of {R}elay {D}eployment for {C}overage {E}xtension in {3GPP} {LTE-A}dvanced,'' in {\em 2009 IEEE international conference on communications workshops}, pp.~1--5, IEEE, 2009.

\bibitem{baccelli2010stochastic}
F.~Baccelli, B.~B{\l}aszczyszyn, {\em et~al.}, ``Stochastic {G}eometry and {W}ireless {N}etworks: {V}olume {II} {A}pplications,'' {\em Foundations and Trends{\textregistered} in Networking}, vol.~4, no.~1--2, pp.~1--312, 2010.

\bibitem{lyu2021hybrid}
J.~Lyu and R.~Zhang, ``Hybrid {A}ctive/{P}assive {W}ireless {N}etwork {A}ided by {I}ntelligent {R}eflecting {S}urface: {S}ystem {M}odeling and {P}erformance {A}nalysis,'' {\em IEEE Transactions on Wireless Communications}, vol.~20, no.~11, pp.~7196--7212, 2021.

\bibitem{wang2023performance}
T.~Wang, G.~Chen, M.-A. Badiu, and J.~P. Coon, ``Performance {A}nalysis of {RIS}-{A}ssisted {L}arge-{S}cale {W}ireless {N}etworks {U}sing {S}tochastic {G}eometry,'' {\em IEEE Transactions on Wireless Communications}, 2023.

\bibitem{sun2023performance}
G.~Sun, F.~Baccelli, K.~Feng, L.~U. Garcia, and S.~Paris, ``A {S}tochastic {G}eometry {F}ramework for {P}erformance {A}nalysis of {RIS}-assisted {OFDM} {C}ellular {N}etworks,'' {\em arXiv preprint arXiv:2310.06754}, 2023.

\bibitem{andrews2011tractable}
J.~G. Andrews, F.~Baccelli, and R.~K. Ganti, ``A {T}ractable {A}pproach to {C}overage and {R}ate in {C}ellular {N}etworks,'' {\em IEEE Transactions on communications}, vol.~59, no.~11, pp.~3122--3134, 2011.

\bibitem{baccelli2020random}
F.~Baccelli, B.~B{\l}aszczyszyn, and M.~Karray, ``Random {M}easures, {P}oint {P}rocesses, and {S}tochastic {G}eometry,'' 2020.

\bibitem{chiu2013stochastic}
S.~N. Chiu, D.~Stoyan, W.~S. Kendall, and J.~Mecke, {\em Stochastic {G}eometry and its {A}pplications}.
\newblock John Wiley \& Sons, 2013.

\bibitem{wu2019intelligent}
Q.~Wu and R.~Zhang, ``Intelligent {R}eflecting {S}urface {E}nhanced {W}ireless {N}etwork via {J}oint {A}ctive and {P}assive {B}eamforming,'' {\em IEEE Transactions on Wireless Communications}, vol.~18, no.~11, pp.~5394--5409, 2019.

\bibitem{zeng2021reconfigurable}
S.~Zeng, H.~Zhang, B.~Di, Y.~Tan, Z.~Han, H.~V. Poor, and L.~Song, ``Reconfigurable {I}ntelligent {S}urfaces in 6{G}: {R}eflective, {T}ransmissive, or {B}oth?,'' {\em IEEE Communications Letters}, vol.~25, no.~6, pp.~2063--2067, 2021.

\bibitem{wang2022reconfigurable}
J.~Wang, Y.-C. Liang, Y.~Pei, and X.~Shen, ``Reconfigurable {I}ntelligent {S}urface as a {M}icro {B}ase {S}tation: {A} {N}ovel {P}aradigm for {S}mall {C}ell {N}etworks,'' {\em IEEE Transactions on Wireless Communications}, vol.~22, no.~4, pp.~2338--2351, 2022.

\bibitem{bogachev2007measure}
V.~I. Bogachev and M.~A.~S. Ruas, {\em Measure {T}heory}, vol.~1.
\newblock Springer, 2007.

\bibitem{alammouri2017sinr}
A.~AlAmmouri, J.~G. Andrews, and F.~Baccelli, ``{SINR} and {T}hroughput of {D}ense {C}ellular {N}etworks with {S}tretched {E}xponential {P}ath {L}oss,'' {\em IEEE Transactions on Wireless Communications}, vol.~17, no.~2, pp.~1147--1160, 2017.

\bibitem{han2020half}
H.~Han, Y.~Liu, and L.~Zhang, ``On {H}alf-power {B}eamwidth of {I}ntelligent {R}eflecting {S}urface,'' {\em IEEE Communications Letters}, vol.~25, no.~4, pp.~1333--1337, 2020.

\bibitem{temme2003large}
N.~M. Temme, ``Large {P}arameter {C}ases of the {G}auss {H}ypergeometric {F}unction,'' {\em Journal of Computational and Applied Mathematics}, vol.~153, no.~1-2, pp.~441--462, 2003.

\bibitem{lukacs1970characteristic}
E.~Lukacs, ``Characteristic {F}unctions,'' 1970.

\bibitem{gakhov2014boundary}
F.~D. Gakhov, {\em Boundary {V}alue {P}roblems}.
\newblock Elsevier, 2014.

\end{thebibliography}

\newpage
\appendix
\subsection{Interchange of Differentiation and Integrals}\label{sec:appendix}
In this technical appendix, we compute the derivative of the ergodic rate $\tau$ in terms of $\lambda_{\rm BS}$ and $\lambda_{\rm RIS}$ in the multiple integral, given by
\begin{equation}\label{eq:background}
\begin{aligned}
	\tau = &\int_{0}^{\infty} \int_{R_c}^{\infty} \bigg( \frac{1}{2\imath \pi} \int_{-\infty}^{\infty} \Big(\mathcal{B}_{\Upsilon}(s-\imath u)- \mathcal{B}_{\Upsilon}(-\imath u) \Big)\frac{{\rm d} u}{ u} + \\ &  \frac{1+\mathcal{B}_{\Upsilon}(s)}{2}\bigg) 
 2\pi \lambda_{\rm BS} r e^{-\pi r^2\lambda_{\rm BS}} \frac{1}{1+t} {\rm d}r {\rm d}t,   
\end{aligned}
\end{equation}
where $s=\frac{1}{P_0g(r)}$ and the Laplace transform of $\Upsilon$ depends on $\lambda_{\rm BS}$ and $\lambda_{\rm RIS}$. 
Recall that $\Upsilon = T(I+\sigma^2)-S_R$ represents the total signal, interference, and Gaussian noise, where $I$ and $S_R$ are both shot noise fields, accounting for fadings across all links and pathloss. 
Given the well-defined distribution functions of both interference and reflected signals as shot noise fields, the probability density function (PDF) $f_{\Upsilon}(t)$ is also well-defined.
The Laplace transform is computed by $\mathcal{B}_{\Upsilon}(s) = \mathcal{B}_{I}(sT)\mathcal{B}_{\sigma^2}(sT) \mathcal{B}_{-S_R}(s)$, given in Eq.~\eqref{eq:laplace_I} and Eq.~\eqref{eq:laplace_SR} 
% $\mathcal{B}_{\sigma^2}(sT) = e^{-sT\sigma^2}$, and 
% \begin{equation}\label{eq:laplace_TI}
%     \mathcal{B}_{I}(sT) = e^{ -2\pi \lambda_{\rm BS} \int_{r}^{\infty} x \big(1-\mathcal{L}_{|\rho_{D}|^2}(s T P_0 g(x))\big){\rm d}x}, 
% \end{equation}
% \begin{equation}\label{eq:laplace_SR}
%     \mathcal{B}_{-S_R}(s) = e^{-\lambda_{\rm RIS} \int_{R_{\rm in}}^{R_{\rm out}} \int_0^{2\pi} y(1 - \mathcal{L}_{|\rho_R|^2}(-s P_0 G(r, y, \psi)  ) ) {\rm d}\psi {\rm d}y }.
% \end{equation}

Prior to interchanging differentiation and integration, we note that the analysis about Eq.~\eqref{eq:background} is meaningful only when the Laplace transform $\mathcal{B}_{\Upsilon}(s)$ is well-defined. 
The following lemma holds: 
\begin{lemma}\label{le:convergence}
For $\mathcal{B}_{\Upsilon}(s)$, the region of convergence is a vertical strip in the complex plane $\Re{s}\in (s_a, s_b)$, where $s_a< 0< \frac{1}{P_0g(r)} < s_b$.
Here,  $s_a<0$ corresponds to the fading of the interference, while $s_b>0$ is associated with the fading of the reflected beam, and their values depend implicitly on the parameters about $\Upsilon$ in Eq.\eqref{eq:laplace_I} and \eqref{eq:laplace_SR}, such as the transmission power $P_0$, the pathloss function $g(r)$, the geometry of RIS deployment area, and the size of RISs. 
\end{lemma}
\begin{proof}
	See \cite[Appendix~B]{sun2023performance}. 
\end{proof}
In addition, we have the following lemma for the PDF $f_{\Upsilon}(t)$ of the random variable $\Upsilon$: 
\begin{lemma}\label{lemma:smoothness}
 The density function
$f_{\Upsilon}(t)$ and its derivative $\frac{\partial f_{\Upsilon}(t)}{\partial t}$ are continuous and bounded, i.e., $\sup_{t}|f_{\Upsilon}(t)|<\infty$, $\sup_{t}\big|\frac{\partial f_{\Upsilon}(t)}{\partial t}\big|<\infty$. 
\end{lemma}
\begin{proof}
We prove the continuity of the density function $f_{\Upsilon}(t)$ through the asymptotic behaviour of the characteristic function $\mathcal{B}_{\Upsilon}(\imath u)$ as $u\rightarrow \pm \infty$. 
We focus on the exponent of Eq.~\eqref{eq:laplace_I}, which can be expressed by its real and imaginary parts: 
\begin{equation}\label{integral_expansion}
\begin{aligned}
	&-2\pi \lambda_{\rm BS} \int_{r}^{\infty} x \big(1-\mathcal{L}_{|\rho_{D}|^2}(\imath u T P_0 g(x))\big){\rm d}x\\
	\stackrel{(a)}{=}& -2\pi \lambda_{\rm BS}\int_{r}^{\infty}
	x\cdot \frac{-\imath u T P_0 g(x)}{1-\imath u T P_0 g(x)}{\rm d}x\\
	=& -2\pi \lambda_{\rm BS}\int_{r}^{\infty}
	x\cdot \bigg( \frac{-\imath u T P_0 g(x)+\big(u T P_0 g(x)\big)^2}{1+ \big(u T P_0 g(x)\big)^2} \bigg){\rm d}x \\
	=& -2\pi \lambda_{\rm BS}\int_{r}^{\infty}
	x\cdot  \frac{\big(u T P_0 g(x)\big)^2}{1+ \big(u T P_0 g(x)\big)^2} {\rm d}x + \\& \qquad \imath \bigg( 2\pi \lambda_{\rm BS}\int_{r}^{\infty}
	x\cdot  \frac{u T P_0 g(x)}{1+ \big(u T P_0 g(x)\big)^2} {\rm d}x \bigg), 
\end{aligned}
\end{equation}
where $(a)$ follows from the characteristic function of an exponential random variable, $\mathcal{L}_{|\rho_{D}|^2}(\imath u T P_0 g(x)) = \frac{1}{1-\imath u T P_0 g(x)}$.
For an exponential function, the real component of the exponent governs its modulus, while the imaginary component determines its oscillatory frequency. 
Consequently, for studying the decay behavior of Eq.~\eqref{eq:laplace_I}, it is sufficient to focus on the real part in Eq.~\eqref{integral_expansion}.
Recall that we adopt a pathloss function $g(x)=\beta(1+x)^{-\alpha}$, where $\beta$ represents the antenna gain for normalization. The real part of interest can be further written as 
\begin{equation}
\begin{aligned}
	&\int_{r}^{\infty}
	x\cdot  \frac{\big(u T P_0 g(x)\big)^2}{1+ \big(u T P_0 g(x)\big)^2} {\rm d}x \\
    =& \int_{r}^{\infty}x\cdot  \frac{\big(u T P_0 \beta(1+x)^{-\alpha}\big)^2}{1+ \big(u T P_0 \beta(1+x)^{-\alpha}\big)^2} {\rm d}x \\
	=& \int_{r}^{\infty}(x+1)\cdot  \frac{\big(u T P_0 \beta(1+x)^{-\alpha}\big)^2}{1+ \big(u T P_0 \beta(1+x)^{-\alpha}\big)^2} {\rm d}x - \\ &\int_{r}^{\infty}  \frac{\big(u T P_0 \beta(1+x)^{-\alpha}\big)^2}{1+ \big(u T P_0 \beta(1+x)^{-\alpha}\big)^2} {\rm d}x \\
	=& \int_{r+1}^{\infty}x\cdot  \frac{\big(u T P_0 \beta x^{-\alpha}\big)^2}{1+ \big(u T P_0 \beta x^{-\alpha}\big)^2} {\rm d}x - \int_{r+1}^{\infty}\frac{\big(u T P_0 \beta x^{-\alpha}\big)^2}{1+ \big(u T P_0 \beta x^{-\alpha}\big)^2} {\rm d}x.
\end{aligned}
\end{equation}
These two integrals can be expressed by Hypergeometric functions, given by
\begin{equation}\label{label:hypergeometric}
\begin{aligned}
&\int_{r+1}^{\infty}
	x\cdot  \frac{\big(u T P_0 \beta x^{-\alpha}\big)^2}{1+ \big(u T P_0 \beta x^{-\alpha}\big)^2} {\rm d}x =    \frac{(u T P_0 \beta)^2 (r+1)^{2-2 \alpha} }{2 (\alpha-1)} \\ &\qquad \, _2F_1\left(1,\frac{\alpha-1}{\alpha};2-\frac{1}{\alpha};-(r+1)^{-2 \alpha} (u T P_0 \beta)^2\right),\\
	&\int_{r+1}^{\infty}\frac{\big(u T P_0 \beta x^{-\alpha}\big)^2}{1+ \big(u T P_0 \beta x^{-\alpha}\big)^2} {\rm d}x = \frac{(u T P_0 \beta)^2 (r+1)^{1-2 \alpha} }{2 \alpha-1}\\
    & \qquad\, _2F_1\left(1,1-\frac{1}{2 \alpha};2-\frac{1}{2 \alpha};-(r+1)^{-2 \alpha} (u T P_0 \beta)^2\right),
	\end{aligned}
\end{equation}
where $_2F_1(,;;)$ is the Hypergeometric function. 
To identify the asymptotic behaviour as $u\rightarrow \pm \infty$, we expand the Hypergeometric function Eq.~\eqref{label:hypergeometric} in terms of series, given by
\begin{equation}\label{hyper_series}
\begin{aligned}
	&\int_{r+1}^{\infty}
	x\cdot  \frac{\big(u T P_0 \beta x^{-\alpha}\big)^2}{1+ \big(u T P_0 \beta x^{-\alpha}\big)^2} {\rm d}x \approx\Big(-{(r+1)^2/}{2}+O\left(u^{-2}\right)\Big) +\\
   & \qquad (u T P_0 \beta)^{2/\alpha} \left(\frac{\pi  \csc \left(\pi/\alpha\right)}{(2 \alpha)}+O\left(u^{-4}\right)\right),  \\
	&  \int_{r+1}^{\infty}\frac{\big(u T P_0 \beta x^{-\alpha}\big)^2}{1+ \big(u T P_0 \beta x^{-\alpha}\big)^2} {\rm d}x \approx \Big(-(r+1)+O\left(u^{-2}\right)\Big) +\\
    & \qquad (u T P_0 \beta)^{1/\alpha} \left(\frac{\pi  \csc \left(\pi /(2 \alpha)\right)}{ (2 \alpha)}+O\left(u^{-4}\right)\right),
\end{aligned}
\end{equation}
where the cosecant function  $\csc(\pi/\alpha)$ is positive for $\alpha>2$, where $\alpha$ represents the pathloss exponent. 
Since $T, P_0, \beta, r$ are constants, the first integral in Eq.~\eqref{hyper_series} tends to $+\infty$ at a rate of $O(u^{2/\alpha})$ while the second grows at a rate of $O(u^{1/\alpha})$ as  $u\rightarrow \pm \infty$~\cite[Eq.~(9)]{temme2003large}, ensuring that the decay behaviour is dominated by the first integral. 
Including the coefficient $-2\pi\lambda_{\rm RIS}$, the real part of Eq.~\eqref{integral_expansion} tends to $-\infty$ at a rate of $O(u^{2/\alpha})$ as $u\rightarrow \pm \infty$. 
In other words, the characteristic function $\mathcal{B}_{I}(\imath u )$ decays at a rate $O(e^{-u^{2/\alpha}})$ as $u\rightarrow \pm \infty$. 
Because the absolute value of the characteristic function of any probability distribution is bounded by one~[Theorem~2.1.1]\cite{lukacs1970characteristic}, the modulus of $\mathcal{B}_{-S_R}(\imath u)$ and $\mathcal{B}_{\sigma^2}(\imath u )$ must be bounded by one. 
Therefore, we conclude that the decay behavior of the modulus of $\mathcal{B}_{\Upsilon}(\imath u)$ is bounded above by that of $\mathcal{B}_{I}(\imath u)$ as $u\rightarrow \pm \infty$ due to the boundedness of the modulus of $\mathcal{B}_{-S_R}(\imath u)$ and $\mathcal{B}_{\sigma^2}(\imath u )$. 

Given this exponential decay rate, $\mathcal{B}_{\Upsilon}(\imath u)$ is absolutely integrable over $(-\infty, \infty)$.
According to~\cite[Theorem~3.2.2]{lukacs1970characteristic}, a characteristic function that is absolutely integrable over $(-\infty, \infty)$ corresponds to an absolutely continuous distribution function, and the Fourier inversion formula expresses the density function of this distribution in terms of the characteristic function, and this density function is continuous. Consequently, the Fourier inversion formula also implies the boundedness of  $f_{\Upsilon}(t)$
\begin{equation}
\begin{aligned}
	\sup_{t}\big|f_{\Upsilon}(t)\big| = \sup_{t} \left| \int_{-\infty}^{\infty}\mathcal{B}_{\Upsilon}(\imath u) e^{\imath ut}\text{d}u \right| &\leq \int_{-\infty}^{\infty}\left|\mathcal{B}_{\Upsilon}(\imath u) \right| \text{d}u\\ &= B_0< \infty.  
\end{aligned}
\end{equation}
where $B_0$ is the limit of the absolute integral, of which the existence is guaranteed by exponential decay.
Furthermore, given a function $f_{\Upsilon}(t)$ with its Fourier transform $\mathcal{B}_{\Upsilon}(\imath u)$, the Fourier transform of its derivative $\frac{\partial f_{\Upsilon}(t)}{\partial t}$ is $\imath u\mathcal{B}_{\Upsilon}(\imath u)$. 
The fast decay property also ensures that the integral $\int_{-\infty}^{\infty}\big|u\mathcal{B}_{\Upsilon}(\imath u) \big|\text{d}u$ is absolutely integrable. Consequently, by the integrable Fourier inversion transform, we have
\begin{equation}\label{eq:differentiable}
\begin{aligned}
	 \sup_{t} \left|  \frac{\partial f_{\Upsilon}(t)}{\partial t}\right|& = \sup_{t} \left|\int_{-\infty}^{\infty}\imath u\mathcal{B}_{\Upsilon}(\imath u) e^{\imath ut}\text{d}u \right|\\
     &\leq \int_{-\infty}^{\infty}\left|u\mathcal{B}_{\Upsilon}(\imath u) \right|\text{d}u = B_1 < \infty,
\end{aligned}
\end{equation}
where $B_1$ is the limit of the absolute integral of $u\mathcal{B}_{\Upsilon}(\imath u)$, of which the existence is also guaranteed by exponential decay. This is because the asymptotic behaviour of the product of a polynomial function and an exponential function is governed by the exponential factor.
\end{proof}

{As already introduced,
this appendix is concerned with the conditions that allow one to interchange the differentiation operator and the integral signs. }
Here, we have the three integrals:
\begin{itemize}
	\item An ordinary integral with respect to the distance $r$ between the UE and the association BS;
	\item Another ordinary integral with respect to the threshold $t$ defining the  coverage success;
	\item A singular integral with respect to the variable $u$, with $u=0$ as a singularity. The integral is defined in the sense of Cauchy principal value $\int_{-\infty}^{\infty} \frac{1}{u}\text{d}u:= \lim_{\epsilon\rightarrow 0} \int_{-\infty}^{-\epsilon}  \frac{1}{u}\text{d}u  + \int^{\infty}_{\epsilon} \frac{1}{u}\text{d}u$ \cite{gakhov2014boundary}. 
\end{itemize}
Since integration regions in our question are unbounded, we use \cite[Corollary 2.8.7]{bogachev2007measure}, which depends on the Lebesgue dominated convergence theorem and the mean value theorem. We rewrite the assertion about the differentiability of integrals here:
\begin{theorem}\label{thm:conditions}
	Let $\mu$ be a nonnegative measure on a space $X$. Let $f:X\times (a,b)\rightarrow \mathbb{R}$ be a function that satisfies the following three conditions:
	\begin{enumerate}
\item For each $t\in (a,b)$, the function $f(x, t)$ is integrable with respect to $x$;
		\item  $f(x, t)$ is differentiable with respect to $t$ for almost every $x\in X$ (with respect to $\mu$);
		\item There exists an $\mu$-a.e. integrable function $\Xi(x)$ such that for every fixed $t$ we have $\big|\frac{\partial f(x, t)}{\partial t}\big|\leq \Xi(x)$. 	
	\end{enumerate}
	Then we have 
	\begin{equation}
		\frac{\partial}{\partial t}\int_X f(x, t)\mu({\rm d}x) = \int_X \frac{\partial f(x, t)}{\partial t}\mu({\rm d}x). 
	\end{equation}
	\end{theorem}
\begin{remark}
	In our case, $t$ can be either $\lambda_{\rm BS}$ or $\lambda_{\rm RIS}$ and $x$ can be the variables for integration. The nonnegative measure $\mu$ is Lebesgue measure and integrability should be understood in the Lebesgue sense.
\end{remark}

Since our aim is to verify whether the integrand functions satisfy the above conditions,   we prove the interchange of integration and differentiation from the inner integral to the outer. In the following, we begin the interchange proof by considering the singular integral and then progress to the two ordinary integrals.  

\subsubsection{Interchange of Differentiation and Singular Integral}

Let the numerator of the integrand of the singular integral be $\Xi_1(s, u, \lambda) := \mathcal{B}_{\Upsilon}(s-\imath u) - \mathcal{B}_{\Upsilon}(-\imath u)$, where $s=\frac{1}{P_0g(r)}$ belongs to the region of convergence of $\mathcal{B}_{\Upsilon}(s)$. 
We use the unspecified variable $\lambda$ to denote either $\lambda_{\rm BS}$ or $\lambda_{\rm RIS}$, since later it will be clear that the analysis applies to both cases. 
In this part, we focus on the behavior of Laplace transforms and singularities, temporarily neglecting the constant factor $\frac{1}{\imath \pi}$. 
Here, we verify the three conditions in Theorem~\ref{thm:conditions}. 
\begin{statement}\label{state_1}
	For every $\lambda$ such that $s$ is within the domain of convergence of the Laplace transform $\mathcal{B}_{\Upsilon}(s)$, the function $\frac{\Xi_1(s, u, \lambda)}{u}$ is integrable w.r.t. Lebesgue measure on $(-\infty, \infty)$. 
\end{statement}
\begin{proof}
Recall that the integral containing the singularity at $u=0$ is defined in the sense of Cauchy principal value.
We first decompose the singular integral $\int_{-\infty}^{\infty}$ into two parts:
\begin{equation}
\begin{aligned} \int_{-\infty}^{\infty}\frac{\Xi_1(s, u, \lambda)}{u} {\rm d}u =& \int_{-1}^{1}\frac{\Xi_1(s, u, \lambda)}{u} {\rm d}u + \\&  \left[\int_{-\infty}^{-1} +
\int_{1}^{\infty}\frac{\Xi_1(s, u, \lambda)}{u} {\rm d}u\right], 
\end{aligned}
\end{equation}
where the first integral is the singular part.

For the first integral, we can express the function within the interval $[-1, 1]$ according to the symmetry as 
\begin{equation}\label{eq:reduce_singularity_to_zero}
\begin{aligned}
    \int_{-1}^{1}\frac{\Xi_1(s, u,\lambda)}{u}{\rm d}u &= \Big( \int_{-1}^{-\epsilon} + \int_{\epsilon}^{1}\Big) \frac{\Xi_1(s, u,\lambda)}{u}{\rm d}u \\ &= \int_{\epsilon}^{1}\frac{\Xi_1(s, u,\lambda)-\Xi_1(s, -u, \lambda)}{u}{\rm d}u.  
\end{aligned}
\end{equation}
%If $\frac{{\rm d}\Xi_1(s, u,\lambda)}{{\rm d}u}$ exists, $u\in[0,1]$, $\frac{\Xi_1(s, u,\lambda)-\Xi_1(s, -u,\lambda)}{u}$ is continuous within the interval $[0, 1]$ and $u=0$ is a removable discontinuity.
Given that the analytic function $\mathcal{B}_{\Upsilon}(s)$ is infinitely differentiable when the Laplace transform is defined within the region of convergence, thus we have $$\lim_{u\rightarrow 0}\frac{\Xi_1(s, u,\lambda)-\Xi_1(s, -u, \lambda)}{u} = 2\cdot \frac{\partial \Xi_1(s, u,\lambda)}{\partial u}.$$ 
This value is finite due to the fact that $s-\imath u$ and $-\imath u$ belong to the domain of convergence. 
In other words, in the region $[0, \epsilon)$ around the singular point, the integrand $\frac{\Xi_1(s, u,\lambda)-\Xi_1(s, -u, \lambda)}{u}$ has a finite value and has a modulus which is bounded above by a constant. 
Thus, $\lim_{\epsilon\rightarrow 0} \int_{0}^{\epsilon}\frac{\Xi_1(s, u,\lambda)-\Xi_1(s, -u, \lambda)}{u}\text{d}u = 0$.
Away from this singular point, the integrand is continuous and bounded and the integration region is finite. 
Hence, $\frac{\Xi_1(s, u,\lambda)-\Xi_1(s, -u, \lambda)}{u}$ is an integrable function defined on the  interval $\lim_{\epsilon\rightarrow0}[\epsilon, 1] = (0, 1]$.

For the second integral, we focus on the decay behavior of $u=\pm \infty$. 
First, reformulate $\Xi_1(s, u,\lambda)$ as
\begin{equation}\label{eq:KUX}
    \Xi_1(s, u,\lambda) = \int_{-\infty}^\infty f_{\Upsilon}(t)(e^{-st}-1)e^{\imath ut} \text{d}t.
\end{equation}
In other words,  $\Xi_1(s, u,\lambda)$ can be seen as the Fourier transform of the function $f_{\Upsilon}(t)(e^{-st}-1)$ with respect to $u$.  
Moreover, the region of convergence of $\mathcal{B}_{\Upsilon}(s)$, as outlined in Lemma~\ref{le:convergence}, is a strip defined by $\Re{s}\in (s_a, s_b)$. 

We introduce two arbitrarily small constant $0<\epsilon_1<\epsilon_2$, and $(s_a+\epsilon_2, s_b-\epsilon_2)\subset [s_a+\epsilon_1, s_b-\epsilon_1]\subset (s_a, s_b)$. 
By $\big|f_{\Upsilon}(t)\big|<\infty$ and the convergence of the Laplace transform for  $s\in [s_a+\epsilon_1, s_b-\epsilon_1]$, we have 
\begin{equation}\label{eq:decay}
\begin{aligned}
&\mathcal{B}_{\Upsilon}(s_a+\epsilon_1)=\int_{-\infty}^{\infty}e^{-(s_a+\epsilon_2) t} f_{\Upsilon}(t) e^{(\epsilon_2-\epsilon_1) t}<\infty \\& \Rightarrow f_{\Upsilon}(t)e^{-(s_a+\epsilon_2)t}=o(e^{-(\epsilon_2-\epsilon_1) t}), \text{ as }t\rightarrow \infty,\\
&\mathcal{B}_{\Upsilon}(s_b-\epsilon_1)=\int_{-\infty}^{\infty}e^{-(s_b-\epsilon_2)t}f_{\Upsilon}(t) e^{(\epsilon_1-\epsilon_2)t}<\infty\\ & \Rightarrow f_{\Upsilon}(t)e^{-(s_b-\epsilon_2)t}=o(e^{(\epsilon_2-\epsilon_1) t}), \text{ as }t\rightarrow -\infty.
\end{aligned}
\end{equation}
By the integrability of $f_{\Upsilon}(t)e^{-st}$ for all $s\in (s_a+\epsilon_1, s_b-\epsilon_1)$ and the boundedness of $f_{\Upsilon}(t)$, as shown in Lemma~\ref{lemma:smoothness}, we have $ f_{\Upsilon}(t)=o\big(e^{(s_a+\epsilon_1)t}\big)$ as $t\rightarrow\infty$, and $f_{\Upsilon}(t)=o\big(e^{(s_b-\epsilon_1)t}\big)$ as $t\rightarrow-\infty$.
If the decay speed is not as fast as in Eq.~\eqref{eq:decay}, but the improper integrals still converge at infinity, the boundedness of $f_{\Upsilon}(t)$ would be violated.

For Eq.~\eqref{eq:KUX} and $s\in [s_a+\epsilon_1, s_b-\epsilon_1]$, $\mathcal{B}_{\Upsilon}(s)< \infty$, so it is possible to apply an integration by parts, 
\begin{equation}\label{different}
\begin{aligned}
	\Xi_1(s, u,\lambda) =& \int_{-\infty}^\infty f_{\Upsilon}(t)(e^{-st}-1)e^{\imath ut} \text{d}t \\=& \Big[f_{\Upsilon}(t)(e^{-st}-1) \cdot \frac{ e^{\imath ut}}{\imath u}\Big]_{-\infty}^{\infty} - \\
    &\int_{-\infty}^\infty \frac{\partial}{\partial t}\Big( f_{\Upsilon}(t)(e^{-st}-1) \Big)\cdot \Big(\frac{ e^{\imath ut}}{\imath u}\Big) \text{d}t \\
	\stackrel{(a)}{=}& 0 - \Big(\frac{1}{ \imath u}\Big) \cdot \int_{-\infty}^\infty \frac{\partial}{\partial t}\Big( f_{\Upsilon}(t)(e^{-st}-1) \Big) e^{\imath ut} \text{d}t.
\end{aligned}
\end{equation}
The step $(a)$ applies to the smaller set $s\in (s_a+\epsilon_2, s_b-\epsilon_2)$,  which ensure the first term to be zero by $f_{\Upsilon}(t)e^{-st}$ decays faster than the growth rate of $e^{-(\epsilon_2-\epsilon_1)t}$ as $t\rightarrow \infty$ and decays faster than $e^{(\epsilon_2-\epsilon_1)t}$ as $t\rightarrow -\infty$, as discussed in Eq.~$\eqref{eq:decay}$. To discuss the absolute integrability of the second term, for $s\neq 0$ (the case $s=0$ is trivial since Eq.~\eqref{different}=0), we can express
\begin{equation}\label{eq:derivative_of_integrand}
\begin{aligned}
	\Big|\frac{\partial}{\partial t}\Big( f_{\Upsilon}(t)(e^{-st}-1) \Big)\Big| &= \Big|\frac{\partial f_{\Upsilon}(t)}{\partial t} (e^{-st}-1) - se^{-st}f_{\Upsilon}(t) \Big| \\
    &\leq \Big|\frac{\partial f_{\Upsilon}(t)}{\partial t} (e^{-st}-1)\Big| + se^{-st}f_{\Upsilon}(t),
\end{aligned}
\end{equation}
where the second term is positive since $s=\frac{1}{P_0g(r)}>0$ so we can drop its absolute value. 
Now we focus on the first term. 
For $s\in [s_a+\epsilon_1,0)\cup(0, s_b-\epsilon_1]$, $\mathcal{B}_{\Upsilon}(s)$ is convergent so it is valid to apply integration by parts, we have 
\begin{equation}\label{eq:derivative_finite}
\begin{aligned}
\mathcal{B}_{\Upsilon}(s)&= \int_{-\infty}^{\infty}f_\Upsilon(t)e^{- st}\text{d}t \\&= \Big[-f_\Upsilon(t)\cdot \frac{1}{ s}e^{- st} \Big]_{-\infty}^{\infty} + \frac{1}{ s} \int_{-\infty}^{\infty}\frac{\partial f_{\Upsilon}(t)}{\partial t} e^{- st}\text{d}t, 
\end{aligned}
\end{equation}
Then, for $s\in (s_a+\epsilon_2,0)\cup(0, s_b-\epsilon_2)$, the first term of Eq.~\eqref{eq:derivative_finite} can be discussed in the same way as the step $(a)$ in Eq.~\eqref{different}, and the value is zero.
Then, the second term of Eq.~\eqref{eq:derivative_finite} is finite and we have $\int_{-\infty}^{\infty}\frac{\partial f_{\Upsilon}(t)}{\partial t} e^{- st}\text{d}t<\infty$.
Recall that by Lemma~\ref{lemma:smoothness}, we have $\frac{\partial f_{\Upsilon}(t)}{\partial t}< \infty$. 
The boundedness and integrability of $\frac{\partial f_{\Upsilon}(t)}{\partial t}$ can establish its decay property based on the same analysis in Eq.~\eqref{eq:decay}, which relies on the same property of $f_{\Upsilon}(t)$. 
That is, the integrability of $\frac{\partial f_{\Upsilon}(t)}{\partial t} e^{-st}$ for $s\in[s_a+\epsilon_1,0)\cup(0, s_b-\epsilon_1]$ allows one to obtain its decay speed for $s\in (s_a+\epsilon_2,0)\cup(0, s_b-\epsilon_2)$, which is faster than $e^{-(\epsilon_2-\epsilon_1)t}$ as $t\rightarrow \infty$ and $e^{(\epsilon_2-\epsilon_1)t}$ as $t\rightarrow -\infty$, respectively. 
The exponential decay speed of $\frac{\partial f_{\Upsilon}(t)}{\partial t} e^{-st}$ at infinity  guarantee it to be absolute integrable, and thus Eq.~\eqref{eq:derivative_of_integrand} is absolute integrable, given by
\begin{equation}
\begin{aligned}
&\int_{-\infty}^{\infty}\Big|\frac{\partial}{\partial t}\Big( f_{\Upsilon}(t)(e^{-st}-1) \Big)\Big| \text{d}t \\&\leq \int_{-\infty}^{\infty}\Big|\frac{\partial f_{\Upsilon}(t)}{\partial t} (e^{-st}-1)\Big|\text{d}t +\frac{1}{s}\mathcal{B}_{\Upsilon}(s) < \infty. 
\end{aligned}
\end{equation}

The absolute integrability of Eq.~\eqref{eq:derivative_of_integrand} implies that
\begin{equation}
	\int_{-\infty}^\infty \frac{\partial}{\partial t}\Big( f_{\Upsilon}(t)(e^{-st}-1) \Big) e^{\imath ut} \text{d}t < \infty.
\end{equation}
Plugging in Eq.~\eqref{different}, we have $|\Xi_1(s, u,\lambda)| = O(\frac{1}{u})$. 
Consequently, $\Big|\frac{\Xi_1(s, u, \lambda)}{u}\Big| = O(\frac{1}{u^2})$ when $u\rightarrow \pm \infty$, thus $\frac{\Xi_1(s, u, \lambda)}{u}$ is bounded by a Lebesgue integrable function in $(-\infty, -1]\cup [1, \infty)$.

In summary, the integrability of $\frac{\Xi_1(s, u,\lambda)}{u}$ is verified by considering both the singularity at $u=0$ and the asymptotic behavior as $u\rightarrow \pm\infty$. 
The singularity at $u=0$  is resolved due to the differentiability of the numerator $\Xi_1(s, u,\lambda)$. 
The asymptotic behavior is guaranteed to decay sufficiently fast whenever $\mathcal{B}_{\Upsilon}(s)$ is defined within the region of convergence.

\end{proof}

\begin{statement}
	Except for the singular point $u=0$, the function $\frac{\Xi_1(s, u, \lambda)}{u}$ is differentiable with respect to $\lambda$.
\end{statement}
\begin{proof}
The set of the singular point $x=\{0\}$ is a null set w.r.t. the Lebesgue measure. Hence the removal of this set does not violate the a.e. qualification of the differentiability.

For all $u\neq 0$, the numerator can be written as 
\begin{equation}\label{eq:first_derivative}
	  \frac{\partial \Xi_1(s, u, \lambda)}{\partial \lambda}  = \mathcal{B}_{\Upsilon}(s-\imath u) \mathcal{D}(s-iu) - \mathcal{B}_{\Upsilon}(-\imath u) \mathcal{D}(-iu), 
\end{equation}
where the abbreviation $\mathcal{D}(s)$ applies for either $\mathcal{D}_{\rm BS}(s)$ or $\mathcal{D}_{\rm RIS}(s)$, respectively. They are derived from Eq.~\eqref{eq:laplace_I} and Eq.~\eqref{eq:laplace_SR}, given by 
\begin{equation}\label{eq:integral}
\begin{aligned}
	\mathcal{D}_{\rm BS}(s):=& -2\pi \int_{r}^{\infty} x \Big(1-\mathcal{L}_{|\rho_{D}|^2}\big(s T P_0 g(x)\big)\Big){\rm d}x, \\
	\mathcal{D}_{\rm RIS}(s):=& -\int_{R_{\rm in}}^{R_{\rm out}} \int_0^{2\pi} y(1 - \mathcal{L}_{|\rho_R|^2}(-s P_0 G(r, y, \psi)  ) ) {\rm d}\psi {\rm d}y. 
\end{aligned}
\end{equation}
The existence of $\mathcal{B}_\Upsilon(s)$ ensures that both $\mathcal{D}_{\rm BS}(s)$ and $\mathcal{D}_{\rm RIS}(s)$ are finite, as the former is essentially the exponential of the latter. The finiteness of the numerator and denominator establishes the differentiability of $\frac{\Xi_1(s, u, \lambda)}{u}$ for $u\neq 0$.
\end{proof}

\begin{lemma}\label{lemma:analytic}
	For $s-\imath u$ within the region of convergence of $\mathcal{B}_{\Upsilon}(s-\imath u)$, $\mathcal{D}_{\rm BS}(s-\imath u)$ and $\mathcal{D}_{\rm RIS}(s-\imath u)$ are differentiable at $u \in 0$.
\end{lemma}
\begin{proof}
We prove the differentiability of $\mathcal{D}_{\rm BS}(s-\imath u)$ at $u=0$. 
Given that the Laplace transform $\mathcal{B}_{I}(s-\imath u)$ is analytic, ensuring the infinite differentiability within its region of convergence,  we can express the first order derivative w.r.t. $u=0$ as 
\begin{equation}\label{eq:differentiability_at_zero}
	\frac{\partial \mathcal{B}_{I}(s-\imath u)}{\partial u}    = \mathcal{B}_{I}(s-\imath u)\cdot \frac{\partial \mathcal{D}_{\rm BS}(s-\imath u)}{\partial u} \lambda_{\rm BS}.
\end{equation}
Hence, 
\begin{equation}
 \frac{\partial \mathcal{D}_{\rm BS}(s-\imath u)}{\partial u} = \frac{\partial \mathcal{B}_{I}(s-\imath u)}{\partial u} \cdot \frac{1}{ \mathcal{B}_{I}(s-\imath u) \lambda_{\rm BS}}.
\end{equation}
Eq.~\eqref{eq:differentiability_at_zero} implies that $\mathcal{D}_{\rm BS}(s-\imath u)$ is differentiable whenever $\mathcal{B}_{I}(s-\imath u) \neq 0$.
In other words, the non-vanishing property of $\mathcal{B}_{I}(s-\imath u)$ ensures that $\mathcal{D}_{\rm BS}(s-\imath u)$ is finite.
Since the PDF $f_{I}(t)$ is positive for $t\in \mathbb{R}^{+}$ and $e^{-st}$ is positive for real $s$ within the region of convergence, i.e., $u=0$, it follows that $\mathcal{B}_{I}(s)=\int_{0}^{\infty}f_{I}(t)e^{-st}\text{d}t$ cannot be zero. 
Hence,  $\frac{\partial \mathcal{D}_{\rm BS}(s-\imath u)}{\partial u}$ is finite at $u=0$, proving the differentiability of $\mathcal{D}_{\rm BS}(s-\imath u)$ w.r.t. $u$ at this point. The proof for $\mathcal{D}_{\rm RIS}(s)$ follows similar steps.
\end{proof}

\begin{statement}\label{statement:differentiability}
	There exists an absolute integrable function that dominates $\Big|\frac{\partial}{\partial \lambda}\frac{\Xi_1(s, u, \lambda)}{u}\Big|$. 
\end{statement}
\begin{proof}
The statement is equivalent to prove that $\frac{\partial}{\partial \lambda}\frac{\Xi_1(s, u, \lambda)}{u}$ is absolutely integrable over $u \in (-\infty, \infty)$. We simply repeat the proof in Statement~\ref{state_1}.
We define $\Xi_2(s, u, \lambda)$ as
\begin{equation}\label{eq:partial_BS_derivative}
     \frac{\partial \Xi_1(s, u, \lambda)}{\partial \lambda} =  \mathcal{B}_{\Upsilon}(s-\imath u) \mathcal{D}(s-\imath u) - \mathcal{B}_{\Upsilon}(-\imath u) \mathcal{D}(-\imath u).
\end{equation}
First, we separate:
\begin{equation}
\begin{aligned} \int_{-\infty}^{\infty}\frac{\Xi_2(s, u, \lambda)}{u} {\rm d}u = &\int_{-1}^{1}\frac{\Xi_2(s, u, \lambda)}{u} {\rm d}u +  \\& \left[\int_{-\infty}^{-1} +\int_{1}^{\infty}\frac{\Xi_2(s, u, \lambda)}{u} {\rm d}u\right].  
\end{aligned}
\end{equation}

For the first integral, using the same step in Eq.~\eqref{eq:reduce_singularity_to_zero}, we have
\begin{equation}
    \int_{-1}^{1}\frac{\Xi_2(s, u,\lambda)}{u}{\rm d}u = \int_{\epsilon}^{1}\frac{\Xi_2(s, u,\lambda)-\Xi_2(s, -u, \lambda)}{u}{\rm d}u.
\end{equation}
At the singularity, we have $\lim_{u\rightarrow 0}\frac{\Xi_2(s, u,\lambda)-\Xi_2(s, -u, \lambda)}{u} = 2\cdot \frac{\partial\Xi_2(s, u, \lambda)}{\partial u}$. 
Given the analyticity of $\mathcal{B}_{\Upsilon}(s)$ and the differentiability of $\mathcal{D}(s-\imath u)$ at $u=0$ (Lemma~\ref{lemma:analytic}), $\frac{\partial\Xi_2(s, u, \lambda)}{\partial u}$ is finite at $u=0$ and has a modulus which is bounded above by a constant.
The differentiability implies that the limit of $\frac{\Xi_2(s, u,\lambda)-\Xi_2(s, -u, \lambda)}{u}$ as $u$ approaches the singularity exists and is finite.
Hence, $\frac{\Xi_2(s, u, \lambda)-\Xi_2(s, -u,\lambda)}{u}$ is a continuous and bounded function in the region containing such a removable singularity.  
Consequently, $\frac{\Xi_2(s, u, \lambda)}{u}$ is absolute integrable over $u\in [-1, -\epsilon)\cup (\epsilon,  1]$.

For the second integral, we check the decay behavior of $\Xi_2(s, u, \lambda)$ at $u=\pm\infty$.
By Lemma~\ref{le:convergence}, we assume that the parameters are defined within the region of convergence of $\mathcal{B}_{\Upsilon}(s)$. 
By writing $\mathcal{B}_{\Upsilon}(s-\imath u) = \int_{-\infty}^{\infty}f_\Upsilon(t)e^{-st+\imath ut}{\rm d}t$ and $\mathcal{B}_{\Upsilon}(-\imath u) = \int_{-\infty}^{\infty}f_\Upsilon(t)e^{\imath ut}{\rm d}t$, the Laplace transforms can be understood as the Fourier transform of $f_\Upsilon(t)e^{-st}$ and $f_\Upsilon(t)$, respectively.
We then can apply the same steps used in Statement~\ref{state_1} to demonstrate that both $\mathcal{B}_{\Upsilon}(s-\imath u)$ and $\mathcal{B}_{\Upsilon}(-\imath u)$ decay at a speed bounded above by $O(\frac{1}{u})$ by the continuity and boundedness of $f_\Upsilon(t)$ and $\frac{\partial f_\Upsilon(t)}{\partial t}$ when $s$ defined within the region of convergence $\mathcal{B}_{\Upsilon}(s)$. 
Hence,  $\Xi_2(s, u, \lambda)$ is the product of vanishing functions $\mathcal{B}_{\Upsilon}(s)$ (at a speed at least of $\frac{1}{u}$) and the exponent functions $\mathcal{D}(s)$ when $u\rightarrow \pm \infty$. 
For the discussion of $\mathcal{D}_{\rm BS}(s)$, recall from Eq.~\eqref{hyper_series} that the exponent of the characteristic function $\mathcal{B}_{\rm BS}(\imath u)$, namely, $\mathcal{D}_{\rm BS}(\imath u)$, exhibits a $O(u^{\frac{2}{\alpha}})$ growth. Including the real part of $s$ in the Laplace transform $\mathcal{B}_{\rm BS}(s)$ will modify the constant term within the denominator of Eq.~\eqref{hyper_series}. However, this modification does not alter the overall growth behavior w.r.t. $u$. 
Considering $\Xi_2(s, u, \lambda)$ as the product of $\mathcal{B}_{\Upsilon}(s)$, which decays faster than $O(u^{-1})$, and $\mathcal{D}_{\rm BS}(s)$ with growth at $O(u^{\frac{2}{\alpha}})$ as $u$ approaches $\pm \infty$, we have $\frac{\Xi_2(s, u, \lambda)}{u}$ decays at a speed $u^{-2+\frac{2}{\alpha}}$ with the pathloss exponent $\alpha>2$.
Therefore, we can establish the absolute integrability of $\frac{\Xi_2(s, u, \lambda)}{u}$ over $u\in (-\infty, -1]\cup[1, \infty)$ for the derivative w.r.t. $\lambda_{\rm BS}$.
For the case of $\mathcal{D}_{\rm RIS}(s)$, where the integration region in Eq.~\eqref{eq:laplace_SR} is compact and the integrand given in Eq.~\eqref{Laplace_transform_rho_R} is bounded. Therefore, $\mathcal{D}_{\rm RIS}(s)$ is bounded and $\mathcal{B}_{\Upsilon}(s)$ decays faster than $O(u^{-1})$, we can also establish the absolute integrability of $\frac{\Xi_2(s, u, \lambda)}{u}$ over $u\in (-\infty, -1]\cup[1, \infty)$ for the derivative w.r.t. $\lambda_{\rm RIS}$.
So both parts of the integral fulfill the criteria for absolute integrability. 
\end{proof}
With the above three statements, we have for both $\lambda_{\rm BS}$ and $\lambda_{\rm RIS}$
\begin{equation}\label{eq:first}
\begin{aligned}
	&\frac{\partial}{\partial \lambda}\int_{-\infty}^{\infty}  \Big(\mathcal{B}_{\Upsilon}(s-\imath u)- \mathcal{B}_{\Upsilon}(-\imath u) \Big)\frac{{\rm d} u}{u} \\ &= 	\int_{-\infty}^{\infty} 	\frac{\partial}{\partial \lambda} \Big(\mathcal{B}_{\Upsilon}(s-\imath u)- \mathcal{B}_{\Upsilon}(-\imath u) \Big)\frac{{\rm d} u}{u}.
\end{aligned}
\end{equation}
Recall that  $\mathcal{B}_{\Upsilon^+}(s)=\int_{-\infty}^{\infty}  \Big(\mathcal{B}_{\Upsilon}(s-\imath u)- \mathcal{B}_{\Upsilon}(-\imath u) \Big)\frac{{\rm d} u}{\imath \pi u} +  \frac{1+\mathcal{B}_{\Upsilon}(s)}{2}$, we have shown the differentiability of $\mathcal{B}_{\Upsilon^+}(s)$. 
\begin{remark}\label{remark_converge}
	By the above statements, we proved that Lemma~\ref{le:convergence} implies the integrability of the integrand function with respect $u$, the differentiability of the integrand function with respect to $\lambda_{\rm BS}$ and $\lambda_{\rm RIS}$, and the integrability of the derivatives with respect to $u$. 
In other words, we have $\mathcal{B}_{\Upsilon^+}(s)$, $\frac{\partial \mathcal{B}_{\Upsilon^+}(s)}{\partial \lambda_{\rm BS}}$, and $\frac{\partial \mathcal{B}_{\Upsilon^+}(s)}{\partial \lambda_{\rm RIS}}$ converge to finite values when $s$ is defined within the region of convergence of $\mathcal{B}_{\Upsilon}(s)$. 
\end{remark}
\subsubsection{Interchange Differentiation and Ordinary Integrals}
Let $\Xi_3(\lambda, r) := \mathcal{B}_{\Upsilon^+}(s)2\pi \lambda_{\rm BS} e^{-\pi r^2  \lambda_{\rm BS}}$. 
\begin{statement}
	For every $\lambda$ defined within the region of convergence of the Laplace transforms $\mathcal{B}_{\Upsilon}(s)$, either $\lambda_{\rm BS}$ or $\lambda_{\rm RIS}$, $\Xi_3(\lambda, r)$ is integrable over $r\in [0, \infty)$.  
\end{statement}
\begin{proof}
	For either $\lambda_{\rm BS}$ or $\lambda_{\rm RIS}$, we have $\Xi_3(\lambda, r)$ is continuous and
	\begin{equation}
		\lim_{r\rightarrow \infty}\mathcal{B}_{\Upsilon^+}(s)2\pi \lambda_{\rm BS} e^{-\pi r^2  \lambda_{\rm BS}} = \lim_{r\rightarrow\infty} O(e^{-r^2}) \stackrel{(a)}{=} 0, 
	\end{equation}
where $(a)$ follows from the fact that all $\mathcal{B}_{\Upsilon^+}(s)$ is finite established in Remark~\ref{remark_converge}. 
\end{proof}

\begin{statement}
For all $r\in [R_c, \infty)$, $\Xi_3(\lambda, r)$ is differentiable with respect to either $\lambda_{\rm BS}$ or $\lambda_{\rm RIS}$. 
\end{statement}
\begin{proof}
	The differentiability of $\mathcal{B}_{\Upsilon^+}(s)$ is established in Statement~\ref{statement:differentiability}.
	The differentiability of $2\pi \lambda_{\rm BS} e^{-\pi r^2  \lambda_{\rm BS}}$ is straightforward. 
	As the multiplication of two differentiable functions, $\Xi_3(\lambda, r)$ is differentiable.
\end{proof}

\begin{statement}
	There exist an integrable function $\Xi_4(\lambda, r)$ dominates $\frac{\partial \Xi_3(\lambda, r)}{\partial \lambda}$
\end{statement}
\begin{proof}
We construct a dominate function $\Xi_4(\lambda, r)$ as 
\begin{equation}
\begin{aligned}
\Big|\frac{\partial \Xi_3(\lambda, r)}{\partial \lambda}\Big| =& \Big|\frac{\partial \mathcal{B}_{\Upsilon^+}(s)}{\partial \lambda}  2\pi \lambda_{\rm BS} e^{-\pi r^2  \lambda_{\rm BS}} +\\
& \mathcal{B}_{\Upsilon^+}(s)\cdot  \frac{\partial}{\partial \lambda}(2\pi \lambda_{\rm BS} e^{-\pi r^2  \lambda_{\rm BS}}) \Big| \\
	\leq & \Big|\frac{\partial \mathcal{B}_{\Upsilon^+}(s)}{\partial \lambda} 2\pi \lambda_{\rm BS} e^{-\pi r^2  \lambda_{\rm BS}}\Big| + \\
    &\Big|\mathcal{B}_{\Upsilon^+}(s)\cdot  \frac{\partial}{\partial \lambda}(2\pi \lambda_{\rm BS} e^{-\pi r^2  \lambda_{\rm BS}}) \Big|. 
\end{aligned}
\end{equation}
Following from all $\mathcal{B}_{\Upsilon^+}(s)$ and $\frac{\partial \mathcal{B}_{\Upsilon^+}(s)}{\partial \lambda}$ are finite established in Remark~\ref{remark_converge} when parameters are defined within the region of convergence of $\mathcal{B}_{\Upsilon}(s)$, $\Xi_4(\lambda, r)$ is integrable since it is continuous w.r.t. $r$ and vanishes when $r=\infty$.
This can be shown by the fact that when $r\rightarrow \infty$, the first part decays at a speed of $O(e^{-r^2})$ for either $\lambda_{\rm BS}$ and $\lambda_{\rm RIS}$, and the second part is zero for $\lambda_{\rm RIS}$ or decays at the same speed for $\lambda_{\rm BS}$.
\end{proof}

\begin{remark}
	Now we established the integrability $\int_{R_c}^{\infty} \Xi_3(\lambda, r){\rm d}r<\infty$, the differentiability $\frac{\partial \Xi_3(\lambda, r)}{\partial \lambda_{\rm BS}}<\infty$ and $\frac{\partial \Xi_3(\lambda, r)}{\partial \lambda_{\rm RIS}}<\infty$, and the integrability of the derivatives  $\int_{R_c}^{\infty}\frac{\partial \Xi_3(\lambda, r)}{\partial \lambda_{\rm BS}}{\rm d}r<\infty$ and $\int_{R_c}^{\infty}\frac{\partial \Xi_3(\lambda, r)}{\partial \lambda_{\rm RIS}}{\rm d}r<\infty$ , when $s$ is defined within the region of convergence of $\mathcal{B}_{\Upsilon}(s)$
\end{remark}
Thus, we have
\begin{equation}\label{eq:second}
\begin{aligned}
&\frac{\partial}{\partial\lambda}\int_{R_c}^{\infty} \mathcal{B}_{\Upsilon^+}(s) 2\pi \lambda_{\rm BS} r e^{-\pi r^2\lambda_{\rm BS}} {\rm d}r \\
&=\int_{R_c}^{\infty}\frac{\partial}{\partial\lambda} \left( \mathcal{B}_{\Upsilon^+}(s) 2\pi \lambda_{\rm BS} r e^{-\pi r^2\lambda_{\rm BS}}\right) {\rm d}r. 
\end{aligned}
\end{equation}
When $\lambda$ specifically refers to $\lambda_{\rm BS}$, we further have 
\begin{equation}
\begin{aligned}
&\frac{\partial}{\partial\lambda_{\rm BS}} \left( \mathcal{B}_{\Upsilon^+}(s) 2\pi \lambda_{\rm BS} r e^{-\pi r^2\lambda_{\rm BS}}\right) = \bigg(\frac{\partial\mathcal{B}_{\Upsilon^+}(s)}{\partial\lambda_{\rm BS}}   \lambda_{\rm BS}  \\
&+\mathcal{B}_{\Upsilon^+}(s) \left(  1-\pi r^2 \lambda_{\rm BS}\right)\bigg) 2\pi re^{-\pi r^2\lambda_{\rm BS}}. 
\end{aligned}
\end{equation}
\begin{statement}
	Let $\Xi_5(\lambda, t) := \int_{0}^\infty \Xi_3(\lambda, r){\rm d}r \frac{1}{1+t}$, the conditions given in Theorem~\ref{thm:conditions} are satisfied.
\end{statement}
\begin{proof}
	By replacing $\mathcal{B}_{\Upsilon^+}(s)$ with $\int_{0}^\infty \Xi_3(\lambda, r){\rm d}r$ and replacing $2\pi \lambda_{\rm BS} e^{-\pi r^2  \lambda_{\rm BS}}$ with $\frac{1}{t+1}$, the verification similar to the last three statements is straightforward.
\end{proof}
Then, we have 
\begin{equation}\label{eq:third}
	\frac{\partial}{\partial\lambda}\int_{0}^\infty \Xi_5(\lambda, t)  {\rm d}t = \int_{0}^{\infty} \frac{\partial}{\partial\lambda} \Xi_5(\lambda, t)  {\rm d}t. 
\end{equation}

Combine Eq.~\eqref{eq:first}, Eq.~\eqref{eq:second} and Eq.~\eqref{eq:third}, we achieve the objective:
\begin{equation}
\begin{aligned}
	\frac{\partial \tau}{\partial\lambda_{\rm BS}} =& \frac{\partial}{\partial\lambda_{\rm BS}}\int_{0}^{\infty} \int_{R_c}^{\infty} \mathcal{B}_{\Upsilon^+}(s) 2\pi \lambda_{\rm BS} r e^{-\pi r^2\lambda_{\rm BS}} \frac{1}{1+t} {\rm d}r {\rm d}t\\
	=&  \int_{0}^{\infty} \int_{R_c}^{\infty}\bigg(  \frac{\partial \mathcal{B}_{\Upsilon^+}(s)}{\partial\lambda_{\rm BS}}  \lambda_{\rm BS}   + \mathcal{B}_{\Upsilon^+}(s) \big( 1 -\pi r^2 \lambda_{\rm BS}\big) \bigg)\\
	&    2\pi r e^{-\pi r^2\lambda_{\rm BS}} \frac{1}{1+t} {\rm d}r {\rm d}t, 
\end{aligned}
\end{equation}
and
\begin{equation}
\begin{aligned}
	\frac{\partial \tau}{\partial\lambda_{\rm RIS}} =& \frac{\partial}{\partial\lambda_{\rm RIS}}\int_{0}^{\infty} \int_{R_c}^{\infty} \mathcal{B}_{\Upsilon^+}(s) 2\pi \lambda_{\rm BS} r e^{-\pi r^2\lambda_{\rm BS}} \frac{1}{1+t} {\rm d}r {\rm d}t\\
	=&  \int_{0}^{\infty} \int_{R_c}^{\infty}  \frac{\partial \mathcal{B}_{\Upsilon^+}(s)}{\partial\lambda_{\rm RIS}} 2\pi \lambda_{\rm BS} r e^{-\pi r^2\lambda_{\rm BS}} \frac{1}{1+t} {\rm d}r {\rm d}t. 
\end{aligned}
\end{equation}

\subsection{Application of the Leibniz rule for Coverage Hole Mitigation}\label{app:leibniz_coverage_hole}

To obtain the derivatives for the coverage hole mitigation scenario, minor modifications to the derivative expressions used for throughput enhancement are required. This is primarily due to the introduction of the exponent's derivative in Eq.~\eqref{eq:derivaive_I}, necessitating the application of the Leibniz rule in interchanging the differentiation and integration in Eq.~\eqref{partial_D_I} and Eq.~\eqref{eq:Derivative_S_R_BS}. 

For applying the Leibniz rule in Eq.~\eqref{partial_D_I} for $\mathcal{D}_{I}(sT)=-2\pi \int_{r_{\rm H}}^{\infty} \frac{xsTP_0g(x)}{1+sTP_0g(x)}{\rm d}x$, it is necessary to demonstrate the continuity and integrability of both the integrand $\frac{xsTP_0g(x)}{1+sTP_0g(x)}$ and its derivative within the interval $[r_{\rm H}, \infty)$. 
The continuity condition is satisfied since the pathloss function, $g(x)=\beta(x+1)^{-\alpha}$, is a polynomial without singularities. Consequently, both the integrand and its derivative, which are expressed as fractions of polynomials without singularities, are also continuous.
The integrability of the integrand, $\frac{xsTP_0\beta(x+1)^{-\alpha}}{1+sTP_0\beta(x+1)^{-\alpha}}$, is ensured by the assumption that the pathloss exponent $\alpha>2$ because the numerator decays faster than $1/x$ due to this pathloss model, and the denominator is always greater than one, ensuring convergence of the integral. This condition also guarantees the integrability of the integrand's derivative, given in Eq.~\eqref{partial_D_I}. 

Since the integration in Eq.~\eqref{eq:Derivative_S_R_BS} is performed over a compact region, interchanging differentiation and integration can be justified when the integrand and its derivative are continuous. This condition is guaranteed given that Eq.~\eqref{eq:derivative_S_R_additional} involve elementary polynomial and exponential functions that are continuous. The only potential singularity arises from the term $1-2 s P_0G(r_{\rm H},y,\psi)M\mathbb{V}[|\zeta|]$. However, $ 1-2 s P_0G(r_{\rm H},y,\psi)M\mathbb{V}[|\zeta|]> 0$ is true due to the condition that the Laplace transform in Eq.~\eqref{Laplace_transform_rho_R} converges.

\end{document}